\documentclass[letterpaper, 11pt]{llncs}

\usepackage[margin=1in]{geometry}
\usepackage[bookmarks, colorlinks=true, plainpages = false, citecolor = blue,linkcolor=red,urlcolor = blue, filecolor = blue]{hyperref}
\usepackage{url}
\usepackage{amsmath,amsfonts,amssymb,bm, verbatim,dsfont,mathtools}
\usepackage{algorithm,algorithmic,color,graphicx,appendix}
\usepackage{subfigure}
\usepackage{etoolbox}
\usepackage{tikz}
\usepackage{xr,xspace}
\usepackage{todonotes}
\usepackage{paralist}
\usepackage{caption}

\usepackage{xspace,prettyref}
\newcommand{\diverge}{\to\infty}

\newcommand{\ones}{\mathbf 1}
\newcommand{\zeros}{\mathbf 0}
\newcommand{\reals}{{\mathbb{R}}}

\newrefformat{eq}{(\ref{#1})}
\newrefformat{chap}{Chapter~\ref{#1}}
\newrefformat{sec}{Section~\ref{#1}}
\newrefformat{algo}{Algorithm~\ref{#1}}
\newrefformat{fig}{Fig.~\ref{#1}}
\newrefformat{tab}{Table~\ref{#1}}
\newrefformat{rmk}{Remark~\ref{#1}}
\newrefformat{clm}{Claim~\ref{#1}}
\newrefformat{def}{Definition~\ref{#1}}
\newrefformat{cor}{Corollary~\ref{#1}}
\newrefformat{lmm}{Lemma~\ref{#1}}
\newrefformat{prop}{Proposition~\ref{#1}}
\newrefformat{app}{Appendix~\ref{#1}}
\newrefformat{hyp}{Hypothesis~\ref{#1}}
\newrefformat{thm}{Theorem~\ref{#1}}

\newcommand{\pth}[1]{\left( #1 \right)}

\newcommand{\norm}[1]{\left\|{#1} \right\|}

\newcommand{\iprod}[2]{\left \langle #1, #2 \right\rangle}

\newcommand{\tx}{{\widetilde{x}}}

\newcommand{\calA}{{\mathcal{A}}}

\newcommand{\calE}{{\mathcal{E}}}
\newcommand{\calF}{{\mathcal{F}}}

\newcommand{\calH}{{\mathcal{H}}}
\newcommand{\calI}{{\mathcal{I}}}

\newcommand{\calN}{{\mathcal{N}}}

\newcommand{\calV}{{\mathcal{V}}}

\newcommand{\argmin}{{\rm argmin}}

\begin{document}

\title{Byzantine Multi-Agent Optimization--Part II\thanks{This research is supported in part by National Science Foundation awards NSF 1329681 and 1421918.
Any opinions, findings, and conclusions or recommendations expressed here are those of the authors
and do not necessarily reflect the views of the funding agencies or the U.S. government.}}

\author{Lili Su \hspace*{1in} Nitin Vaidya}
\institute{Department of Electrical and Computer Engineering, and\\
Coordinated Science Laboratory\\
University of Illinois at Urbana-Champaign\\
Email:\{lilisu3, nhv\}@illinois.edu}
\maketitle

\begin{center}
Technical Report\\
~\\

July 2015
\end{center}
~

\centerline{\bf Abstract}

In Part I of this report, we introduced a Byzantine fault-tolerant distributed optimization problem whose goal is to optimize a sum of convex (cost) functions with real-valued scalar
input/ouput. In particular, the goal is to optimize
a global cost function $\frac{1}{|\calN|}\sum_{i\in \calN} h_i(x)$, where $\calN$ is the set of non-faulty
agents, and $h_i(x)$ is agent $i$'s local cost function, which is initially
known only to agent $i$. In general,
when some of the agents may be Byzantine faulty, the above goal is unachievable. 
Therefore, in Part I, we studied a weaker version of the problem whose goal is to generate an output that is an optimum of a function formed as
a {\em convex combination} of local cost functions of the non-faulty agents.
We showed that the maximum achievable number of weights ($\alpha_i$'s) that are bounded away from 0 is $|\calN|-f$, where $f$ is the upper bound on the number of Byzantine agents.\\

In this second part, we introduce a condition-based variant of the original problem over arbitrary directed graphs. Specifically, for a given collection of $k$ input functions $h_1(x), \ldots, h_k(x)$, we consider the scenario when the local cost function stored at agent $j$, denoted by $g_j(x)$, is formed as a {\em convex combination} of the $k$ input functions $h_1(x), \ldots, h_k(x)$. The goal of this condition-based problem is to generate an output that is an optimum of $\frac{1}{k}\sum_{i=1}^k h_i(x)$.
Depending on the availability of side information at each agent, two slightly different variants are considered.  We show that for a given graph, the problem can indeed be solved despite the presence of faulty agents. In particular, even in the absence of side information at each agent, when adequate {\em redundancy} is available in the optima of input functions, a distributed algorithm is proposed in which each agent carries minimal state across iterations.

\begin{keywords}
Distributed optimization; Byzantine faults; incomplete networks; fault-tolerant computing
\end{keywords}


\section{System Model and Problem Formulation}
The system under consideration is synchronous, and consists of $n$ agents connected by an arbitrary directed communication network $G(\calV,\calE)$, where $\calV=\{1,\dots,n\}$ is the set of $n$ agents, and $\calE$ is the set of directed edges between the agents in $\calV$. Up to $f$ of the $n$ agents may be Byzantine faulty.
Let $\calF$ denote the set of faulty agents in a given execution. Agent $i$ can reliably transmit messages to agent $j$ if and only if
the directed edge $(i,j)$ is in $\calE$.
Each agent can send messages to itself as well, however,
for convenience, we {\em exclude self-loops} from set $\calE$.
That is, $(i,i)\not\in\calE$ for $i\in\calV$.
With a slight abuse of terminology, we will use the terms {\em edge}
and {\em link} interchangeably, and use the terms {\em nodes}
and {\em agents} interchangeably in our presentation.

For each agent $i$, let $N_i^-$ be the set of agents from which $i$ has incoming
edges.
That is, $N_i^- = \{\, j ~|~ (j,i)\in \calE\, \}$.
Similarly, define $N_i^+$ as the set of agents to which agent $i$
has outgoing edges. That is, $N_i^+ = \{\, j ~|~ (i,j)\in \calE\, \}$.
Since we exclude self-loops from $\calE$,
$i\not\in N_i^-$ and $i\not\in N_i^+$.
However, we note again that each agent can indeed send messages to itself.
Agent $j$ is said to be an {\em incoming neighbor} of agent $i$,
if $j\in N_i^-$. Similarly, $j$ is said to be an {\em outgoing neighbor}
of agent $i$, if $j\in N_i^+$.\\

We say that a function $h: \mathbb{R}\rightarrow \mathbb{R}$ is {\em admissible} if (i) $h(\cdot)$ is convex and $L$-Lipschitz continuous,
and (ii) the set $\argmin~h(x)$ containing the optima of $h(\cdot)$
is non-empty and compact (i.e., bounded and closed).
Given $k$ admissible input functions $h_1(x), \ldots, h_k(x)$, each agent $i\in \calV$ is initially provided with a local cost function $g_i(\cdot)$ of the form
\begin{align*}
g_i(x)={\bf A}_{1i}h_1(x)+{\bf A}_{2i}h_2(x)+\ldots +{\bf A}_{ki}h_k(x),
\end{align*}
where ${\bf A}_{ji}\ge 0$ and ${\bf A}_{ji}\ge 0$ and $\sum_{j=1}^k {\bf A}_{ji}=1$ for all $i\in \calV$ and all $j=1, \ldots, k$. Compactly, we have ${\bf g}(x)={\bf h}(x){\bf A}$, where ${\bf h}(x)=[h_1(x), h_2(x), \ldots, h_k(x)]$, ${\bf g}(x)=[g_1(x), g_2(x), \ldots, g_n(x)]$ and ${\bf A}\in \reals^{k\times n}$. Our problem formulation is motivated by the work on condition-based consensus \cite{mostefaoui2003using,Chaudhuri92morechoices,friedman2007asynchronous}, where the inputs of the agents are restricted to be within some acceptable set.

Each agent $i$ maintains state $x_i$, with $x_i(t)$ denoting the local estimate of the
optimal $x$, computed by node $i$ at the {\em end}\, of the $t$-th iteration of the algorithm,
with $x_i(0)$ denoting its initial local estimate.
At the {\em start} of the $t$-th iteration ($t>0$), the local estimate of
agent $i$ is $x_i(t-1)$. The algorithms of interest will require each agent $i$
to perform the following three steps in iteration $t$, where $t>0$.
Note that the faulty agents may deviate from this specification.
Since each $h_j(\cdot)$ is convex and $L$-Lipschitz continuous, and $\sum_{j=1}^k {\bf A}_{ji}=1$, it follows that each $g_i(\cdot)$ is also convex and $L$-Lipschitz continuous. Note that the formulation allows $n<k$ as well as $n\geq k$. The matrix $\bf A$ is termed as a {\em job assignment matrix}.
The goal here is to develop algorithms that output $x_i=\tx$ at each non-faulty agent $i$ such that
\begin{align}
\label{goal}
\tx\in \argmin~ h(x)=\frac{1}{k}\sum_{j=1}^{k} h_j(x).
\end{align}
That is, we are interested in developing algorithms in which the local estimate of each non-faulty agent will eventually reach consensus, and the consensus value is an optimum of function $h(\cdot)$.

Let $X_j=\argmin ~h_j(x)$ for all $j=1, \ldots, k$, and let $X=\argmin ~h(x)$.
For ease of future reference, we refer to the above optimization problem \ref{goal} as Problem (\ref{goal}). Problem \ref{goal} is said to be solvable if there exists an algorithm that outputs $\tx\in \argmin~ h(x)$ at each non-faulty agent $i$ for any collection of $k$ admissible functions. 
Problem \ref{goal} can be further formulated differently depending on whether each non-faulty agent $i$ knows the assignment matrix $\bf A$ or not.
We refer to the formulation where the agents know matrix $\bf A$ as condition-based Byzantine multi-agent optimization \textbf{with} side information; otherwise the problem is called condition-based Byzantine multi-agent optimization \textbf{without} side information.\\

Our formulation is more general than the common formulation adopted in \cite{Duchi2012,nedic2015distributed,Nedic2009,ram2010distributed,Tsianos2012,tsitsiklis1986distributed}, in which $f=0$ and the assignment matrix ${\bf A}={\bf I}_k$ (identity matrix) is considered. Despite the elegance of the algorithms proposed in \cite{Duchi2012,nedic2015distributed,Nedic2009,ram2010distributed,Tsianos2012,tsitsiklis1986distributed}, none of these algorithms work in the presence of Byzantine agents when $f\ge 1$ and ${\bf A}={\bf I}_k$. Informally speaking, this is because under ${\bf g}(x)={\bf h}(x){\bf I}_k={\bf h}(x)$ assignment, the information about the input function $h_i(x)$ is exclusively known to agent $i$ in the system. If agent $i$ is faulty and misbehaves, or crashes at the beginning of an execution, then the information about $h_i(x)$ is not accessible to the non-faulty agents. When $\sum_{j=1}^k h_j(x)$ and $\sum_{j=1, j\not=i}^k h_j(x)$ do not have common optima, there does not exist a correct algorithm.
A stronger impossibility result is presented next, which is proved in Part I of our work \cite{SuV15MultiAgentPartI}.

\begin{theorem}\cite{SuV15MultiAgentPartI}
\label{t_imposs_0}
Problem \ref{goal} is not solvable when $f\ge 1$ and ${\bf A}={\bf I}_k$.
\end{theorem}
In contrast, function redundancy can be added to the system by applying a properly chosen job assignment matrix ${\bf A}$ to ${\bf h}(x)$. For example, suppose $k=2$, $f=1$ and the optimal sets of functions $h_1(x)$ and $h_2(x)$ are $[-1, 0]$ and $[0, 1]$, respectively. Let ${\bf g}(x)={\bf h}(x) {\bf G}$, where ${\bf G}$ is a generator matrix of a repetition code with $d=2f+1=3$.
Informally speaking, by applying linear code $\bf G$ on input functions ${\bf h}(x)$, i.e., ${\bf g}(x)={\bf h}(x){\bf G}$, the Byzantine agents' ability in hiding information about input functions can be weakened. This observation and Theorem \ref{t_imposs_0} together justify our problem formulation.

\paragraph{{\bf Contributions:}}

We introduce a condition-based approach to Byzantine multi-agent optimization problem. Two slightly different variants are considered: condition-based Byzantine multi-agent optimization with side information and condition-based Byzantine multi-agent optimization without side information. For the former, when side information is available at each agent, a decoding-based algorithm is proposed, assuming that each input function is differentiable. This algorithm combines the gradient method with the decoding procedure introduced in \cite{candes2005decoding} (namely matrix ${\bf A}$). With such a decoding subroutine, our algorithm essentially performs the gradient method, where gradient computation is performed distributedly over the multi-agent system.
When side information is not available at each agent, we propose a simple consensus-based algorithm in which each agent carries minimal state across iterations. This consensus-based algorithm solves Problem \ref{goal} under the additional assumption over input functions that all input functions share at least one common optimum. 

\paragraph{{\bf Organization:}}
The rest of the report is organized as follows. Related work is summarized in Section \ref{sec: related work}.
Condition-based Byzantine multi-agent optimization with side information is analyzed in Section \ref{sec:SolutionIndependent}, where each agent knows the assignment matrix ${\bf A}$. Section \ref{sec:SolutionRedundant} is devoted to the case when each agent does not know ${\bf A}$. Section \ref{sec:summary} concludes the report.

\section{Related Work}\label{sec: related work}

Fault-tolerant consensus \cite{PeaseShostakLamport} is closely related to the optimization problem considered in this report. There is a significant body of work on fault-tolerant consensus, including \cite{Dolev:1986:RAA:5925.5931,Chaudhuri92morechoices,mostefaoui2003conditions,fekete1990asymptotically,LeBlanc2012,vaidya2012iterative,friedman2007asynchronous}.
Two variants that are most relevant to the algorithms in this report are {\em iterative approximate Byzantine consensus} \cite{fekete1990asymptotically,LeBlanc2012,vaidya2012iterative} and {\em condition-based consensus} \cite{mostefaoui2003using,Chaudhuri92morechoices,friedman2007asynchronous}. Iterative approximate consensus requires that the agents agree with each other only approximately, using local communication and maintaining {\em minimal} state across iterations. Condition-based consensus \cite{mostefaoui2003using} restricts the inputs of the agents to be within some acceptable set. \cite{Chaudhuri92morechoices} showed that if a condition (the set of allowable system inputs) is $f$--{\em acceptable}, then consensus can be achieved in the presence of up to $f$ crash failures over complete graphs.
 A connection between asynchronous consensus and error-correcting codes (ECC) was established in \cite{friedman2007asynchronous}, observing that crash failures and Byzantine failures correspond to erasures and substitution errors, respectively, in ECCs. Condition-based approach can also be used in synchronous system to speed up the agreement \cite{Mostefaoui2006,mostefaoui2003using}. 

Convex optimization, including distributed convex optimization, also has a long history \cite{bertsekas1989parallel}. 
Primal and dual decomposition methods that lend themselves naturally to a distributed paradigm are well-known \cite{Boyd2011}. 
There has been significant research on a variant of distributed optimization problem \cite{Duchi2012,nedic2015distributed,Nedic2009,ram2010distributed,Tsianos2012,tsitsiklis1986distributed}, in which the global objective $h(x)$ is a summation of $n$ convex functions, i.e, $h(x)=\sum_{j=1}^n h_j(x)$, with function $h_j(x)$ being known to the $j$-th agent. The need for robustness for distributed optimization problems has received some attentions recently \cite{Duchi2012,nedic2015distributed,ram2010distributed}.
In particular, Ram et al.\ \cite{ram2010distributed} studied the scenario when each component function is known partially (with stochastic errors) to an agent, Duchi et al.\ \cite{Duchi2012} and Nedic et al.\ \cite{nedic2015distributed} investigated the impact of random communication link failures and time-varying communication topology.  Duchi et al.\ \cite{Duchi2012}  assumed that each realizable link failure pattern considered in \cite{Duchi2012} is assumed to admit a doubly-stochastic matrix which governs
the evolution dynamics of local estimates of the optimum. The doubly-stochastic requirement is relaxed in \cite{nedic2015distributed}, using the push-sum technique used in \cite{Tsianos2012}.  In contrast, we consider the system in which up to $f$ agents may be Byzantine, i.e., up to $f$ agents may be adversarial and try to mislead the system to function improperly. We are not aware of the existence of results obtained in this report.

In other related work, significant attempts have been made to solve the problem of distributed hypothesis testing in the presence of Byzantine attacks \cite{kailkhura2015consensus,zhang2014distributed,marano2009distributed}, where Byzantine sensors may transmit fictitious observations aimed at confusing the decision maker to arrive at a judgment that is in contrast with the true underlying distribution. Consensus based variant of distributed event detection, where a centralized data fusion center does not exist, is considered in \cite{kailkhura2015consensus}. In contrast, in this paper, we focus on the Byzantine attacks on the multi-agent optimization problem.

\section{Condition-based Byzantine multi-agent optimization with side information}
\label{sec:SolutionIndependent}
In this section we consider condition-based Byzantine multi-agent optimization with side information, where each agent knows the assignment matrix ${\bf A}$.  Let $\{\alpha(t)\}_{t=0}^{\infty}$ be a sequence of step sizes. A simple decoding-based algorithm, Algorithm 1, formally presented below, works in an iterative fashion. Recall that $x_i(0)$ is the initial state of local estimate for each non-faulty agent $i\in \calV-\calF$, and $G(\calV, \calE)$ is the underlying communication graph. Without loss of generality, we assume that $x_i(0)=x_0$ for $i\in \calV-\calF$ and some arbitrary but fixed $x_0\in \reals$. Otherwise,  we can add an additional initialization step to guarantee identical  ``initial state" using an arbitrary exact consensus algorithm.
Let $x_i(t)$ be the local estimate of an
optimum in $X$, computed by node $i$ at the {\em end}\, of the $t$-th iteration of the algorithm.
At the {\em start} of the $t$-th iteration ($t>0$), the local estimate of
agent $i$ is $x_i(t-1)$.

For Algorithm 1 to work, we assume that each input function $h_i(\cdot)$ is differentiable. Consequently, the local objective $g_i(\cdot)$ is also differentiable for each $i\in \calV$. Let ${\bf A}\in \reals^{k\times n}$ be a matrix that can corrects up to $f$ arbitrary entry-wise errors in \cite{candes2005decoding}.
At iteration $t$, each non-faulty agent $i$ computes the gradient of $g_i(t)$ at $x_i(t-1)$. Let ${\bf d}(t)$ be the $k$-dimensional vector of the gradients of the $k$ input functions at $x_i(t-1)$, where $i\in \calV-\calF$. For the $j$--th entry in ${\bf d}(t)$, i.e, ${\bf d}_j(t)$,  it holds that ${\bf d}_j(t)=h^{\prime}_j(x_i(t-1))$.   Later we will show that $x_i(t-1)=x_j(t-1)$ for all $i, j\in \calV-\calF$. Thus ${\bf d}(t)$ is well-defined.  In addition, we assume the structure of the underlying graph $G(\calV, \calE)$ admits Byzantine broadcast. For instance, when $G(\calV, \calE)$ is undirected, for a correct Byzantine broadcast algorithm to exist, node connectivity of $G(\calV, \calE)$ is at least $2f+1$.

\paragraph{}
\vspace*{8pt}\hrule

{\bf Algorithm 1: }
\vspace*{4pt}\hrule

~

Steps to be performed by agent $i\in \calV$ in iteration $t\ge 0$.\\

{\bf Initialization: } $x_i(0)\gets x_0$.

\begin{enumerate}

\item {\em Transmit step:} Compute $g_i^{\prime}(x_i(t-1))$,  the gradient of $g_i(\cdot)$ at $x_i(t-1)$, and perform Byzantine broadcast of $g_i^{\prime}(x_i(t-1))$ to all agents.\\

\item {\em Receive step:} Receive gradients from all other agents. Let ${\bf y}^i(t)$ be a $n$--dimensional vector of received gradients, with ${\bf y}^i_j(t)$ being the value received from agent $j$. If $j\in \calV-\calF$, then ${\bf y}^i_j(t)=g_j^{\prime}(x_j(t-1))$.\\

\item {\em Gradient Decoding step:}
Perform the decoding procedure in \cite{candes2005decoding} to recover
$${\bf d}(t)=[h_1^{\prime}(x_i(t-1)), \cdots, h_k^{\prime}(x_i(t-1))]^{T}.$$

\item \textit{Update step:} Update its local estimate as follows.
    \begin{align}
    \label{update decoding}
    x_i(t)=x_i(t-1)-\alpha(t-1) \sum_{j=1}^k h_j^{\prime}(x_i(t-1)).
    \end{align}

\end{enumerate}
\hrule

~

At iteration $t=1$, each non-faulty agent $i$ computes $g_i^{\prime}(x_i(0))$--the gradient of $g_i(\cdot)$ at the current estimate $x_i(0)=x_0$, and performs Byzantine broadcast of $g_i^{\prime}(x_0)$.  Note that a faulty agent $p$, instead of $g_p^{\prime}(x_0)$, may perform Byzantine broadcast of some arbitrary value to other agents.
Recall that ${\bf y}^i(1)\in \reals^n$ is a $n$--dimensional real vector, with ${\bf y}^i_j(1)$ be the value received from agent $j$ at iteration $1$.  Since ${\bf g}(\cdot)={\bf h}(\cdot){\bf A}$, then we can write ${\bf y}^i(1)$ as ${\bf y}^i(1)={\bf d}(1) {\bf A}+{\bf e}^i(1)$, where ${\bf e}^i(1)$ corresponds to the errors induced by the faulty agents. Let $p$ be a nonzero entry in ${\bf e}^i(1)$, it should be noted that ${\bf e}_p^i(1)$ can be  arbitrarily away from 0.  Since messages/values are transmitted via Byzantine broadcast, it holds that ${\bf e}^i(1)={\bf e}^{i^{\prime}}(1)$ for all $i, i^{\prime}\in \calV-\calF$. Consequently, we have ${\bf y}^i(1)={\bf y}^{i^{\prime}}(1)$ for all $i, i^{\prime}\in \calV-\calF$. For each $i\in \calV-\calF$, ${\bf d}(1)$ can be recovered using the decoding procedure in \cite{candes2005decoding}. By the updating function
(\ref{update decoding}), we know $x_i(1)=x_j(1)$ for all $i,j \in \calV-\calF$. Inductively, it can be shown that $x_i(t)=x_j(t)$ for all $i,j \in \calV-\calF$, and for all $t\ge 0$. Thus, ${\bf d}(t)$ is well-defined for each $t\ge 0$.
%
%
%
The remaining correctness proof of Algorithm 1 follows directly from the standard gradient method convergence analysis for convex objective.

Due to the use of Byzantine broadcast, the communication load in Algorithm 1 is high. The communication cost can be reduced by using a matrix ${\bf A}$ that has stronger error-correction ability. 
In general, there is some tradeoff among the communication cost, the graph structure and the error-correcting capability of ${\bf A}$. Our main focus of this paper is the case when no side-information is available at each agent, thus we do not pursue this tradeoff further.


\section{Condition-based Byzantine multi-agent optimization without side information}
\label{sec:SolutionRedundant}
In this section, we consider the scenario when side information about the assignment matrix ${\bf A}$ is not known to each agent.
We will classify the collection of input functions into three classes depending on the level of redundancy in the input function solutions.
For functions with adequate redundancy in their optima,  a simple consensus-based algorithm, named Algorithm 2, is proposed. Although Algorithm 2, at least in its current form, only works for a restricted class of input functions, it is more efficient in terms of both memory and local computation, compared to Algorithm 1. We leave the adaptation of Algorithm 2 to the general input functions as future work.

The job assignment matrices used in this section are characterized by {\em sparsity parameter}--a new property (introduced in this report) over matrices. 

\subsection{Classification of input functions collections}
Recall that protective function redundancy is added to the system by applying a proper matrix ${\bf A}$ to ${\bf h}(\cdot)$, i.e., ${\bf g}(\cdot)={\bf h}(\cdot){\bf A}$. In Algorithm 1 sufficient redundancy is added to the system such that Algorithm 1 works for any collection of input functions.  However, for some collection of input functions, such function redundancy may not be necessary.  Consider the case when all $k$ input functions are strictly convex and have the same optimum, i.e., $X_i=\{x^*\}$ for some $x^*$ and for all $i=1, \ldots, k$. In addition, the agents know that $X_i=X_j$ and $|X_i|=1$ for all $i,j\in \calV$. It can be checked that $h(x)=\frac{1}{k}\sum_{j=1}^k h_j(x)$ is also strictly convex and $X=\{x^*\}$.
Even if there is no redundant agents in the system and no redundancy added when applying $\bf A$, i.e., ${\bf A}={\bf I}_k$, Problem \ref{goal} can be solved trivially by requiring each non-faulty agent to minimize its own local objective $h_i(x)$ individually without exchanging any information with other agents.

Informally speaking, as suggested by the above example, the optimal sets of the given input functions may themselves have redundancy. For ease of further reference, we term this redundancy as \emph{solution redundancy}. Closer examination reveals that the collections of input functions can be categorized into three classes according to solution redundancy.
\begin{itemize}
\item[Case 1:] The $k$ input functions are strictly convex and $X_i=\{x^*\}$ for all $i=1, \ldots, k$, and the agents know that $X_i=X_j$ and $|X_i|=1$ for all $i,j\in \calV$;
\item[Case 2:] The $k$ input functions share at least one common optimum, i.e., $\cap_{i=1}^k X_i\not=\O$, and the agents know that $\cap_{i=1}^k X_i\not=\O$;
\item[Case 3:] The $k$ input functions share no optima, i.e., $\cap_{i=1}^k X_i=\O$.
\end{itemize}
If the collection of $k$ input functions belongs to Case 1 or Case 2, we refer to this scenario as \emph{solution-redundant} functions; similarly, we refer to the collection of $k$ functions that falls within Case 3 as \emph{solution-independent} functions.
When the given collection of input functions fits Case 2 or Case 3 (but not Case 1), information exchange among agents is in general required in order to achieve asymptotic consensus over local estimates of non-faulty agents.\\

In this section, we are particularly interested in the family of algorithms of the following structure.

\subsection{Algorithm Structure}
\label{subsec:algorithm}
Recall that each agent $i$ maintains state $x_i$, with $x_i(t)$ denoting the local estimate of an
optimum in $X$, computed by node $i$ at the {\em end}\, of the $t$-th iteration of the algorithm,
with $x_i(0)$ denoting its initial local estimate.
At the {\em start} of the $t$-th iteration ($t>0$), the local estimate of
agent $i$ is $x_i(t-1)$. The algorithms of interest will require each agent $i$
to perform the following three steps in iteration $t$, where $t>0$.
Note that the faulty agents may deviate from this specification.

\begin{enumerate}
\item {\em Transmit step:} Transmit message $m_i(t)$ on all outgoing edges (to agents in $N_i^{+}$).
\item {\em Receive step:} Receive messages on all incoming edges (from agents in $N_i^{-}$). Denote by $r_i(t)$ the vector of messages received from its neighbors.
\item \textit{Update step:} Agent $i$ updates its local estimate using a transition function $Z_i$, 
\begin{eqnarray}
x_i(t)=Z_i\pth{r_i(t), x_i(t-1), g_i(\cdot)},
\label{eq:Z_i}
\end{eqnarray}
where $Z_i$ is a part of the specification of the algorithm.
\end{enumerate}

The evolution of local estimate at agent $i$ is governed by the update function defined in (\ref{eq:Z_i}). Note that $x_i(t)$ only depends on local objective $g_i(\cdot)$, $x_i(t-1)$ and $r_i(t)$--the messages collected by agent $i$ in the receive step of iteration $t$. No other information collected in any of the previous iteration will affect the update step in iteration $t$. Intuitively speaking, non-faulty agent $i$ is assumed to have no memory across iterations except $x_i$. Note that the information available at each non-faulty node $i\in \calV-\calF$ is the local estimate $x_i(t-1)$ and the local objective $g_i(\cdot)$. Thus, the message $m_i(t)$ is a function of $x_i(t-1)$ and $g_i(\cdot)$ only, i.e., $$m_i=F_i(x_i(t-1), g_i(\cdot)).$$\\

An algorithm is said to be correct (1) if $\lim_{t\diverge}|x_i(t)-x_j(t)|=0 ~ \text{and}~ \lim_{t\diverge} x_j(t)\in X,~\text{for all initial states}~ x_j(0)~\text{and for all}~i,j\in \calV-\calF,$
and (2) if there exists a finite $t_0$ such that $x_i(t_0)=x_j(t_0)$ and $x_i(t_0)\in X$ for all $i,j\in \calV-\calF$, then
$x_i(t)=x_i(t_0)~\text{for all}~ i\in \calV-\calF  ~\text{and for all}~t\ge t_0.$

Case 1 above is a special form of Case 2. For Case 1, where $h_j(x)$'s are strictly convex and have the same optimum, the problem can be solved trivially. However, for Case 2 in general, the redundancy that is necessary
 may depend on the underlying graph structure. Henceforth, we consider the scenario when the input functions falls in Case 2. Note that Theorem \ref{t_imposs_0} still holds when restricting to Case 2 input functions.
Next, we introduce the notion of sparsity parameter of a job assignment matrix, and characterize the tradeoff between the sparsity parameter and the necessary and sufficient condition, for a correct algorithm to exist.

\begin{definition}
Given a job assignment matrix $\bf A$, the sparsity parameter of $\bf A$, denoted by $sp({\bf A})$, is the smallest integer such that the sum vector of any $sp({\bf A})$ columns of $\bf A$ is component-wise positive, i.e., every coordinate of the sum vector is positive. In particular, if the sum vector of all columns of $\bf A$ is not component-wise positive, then $sp({\bf A})\triangleq n+1$ by convention. \end{definition}
Recall that ${\bf A}\ge \zeros$ is a nonnegative matrix, $sp({\bf A})\triangleq n+1$ implies that there exists a row in $\bf A$ that contains only zeros. The following lemma presents a lower bound on the number of nonzero elements in a row of $\bf A$, given that $sp({\bf A})=k^{\prime}$.

\begin{lemma}
\label{sparsity}
Given an assignment matrix $\bf A$, its sparsity parameter $sp({\bf A})=k^{\prime}$ if and only if there are at most $k^{\prime}-1$ zero entries in each row of $\bf A$ and there exists one row that contains exactly $k^{\prime}-1$ zero entries.
\end{lemma}
Lemma \ref{sparsity} is proved in Appendix \ref{app:SRNC}.

The sparsest assignment matrix $\bf A$ with $sp({\bf A})=k^{\prime}$ can be constructed by choosing arbitrary $k^{\prime}-1$ entries in each row to be zero. By the proof of Lemma \ref{sparsity}, it can be checked that the sparsity parameter of the obtained matrix $\bf A$ is $k^{\prime}$. In addition, the total number of non-zero entries in $\bf A$ is $\pth{n-k^{\prime}+1}k$.


\subsection{Terminology of Consensus}
\label{sec:necessary}

Our condition is based on characterizing a special of subgraphs of $G(\calV, \calE)$, termed by reduced graph \cite{vaidya2012iterative}, formally defined below.
\begin{definition}\cite{vaidya2012iterative}
For a given graph $G(\calV, \calE)$, a reduced graph $\calH$ is a subgraph of $G(\calV, \calE)$ obtained by (i) removing all the faulty agents from $\calV$ along with their edges; (ii) removing any additional up to $f$ incoming edges at each non-faulty agent.
\end{definition}
Let us denote the collection of all the reduced graphs for a given $G(\calV, \calE)$ by $R_\calF$. Thus, $\calV-\calF$ is the set of agents in each element in $R_\calF$. Let $\tau=|R_\calF|$. It is easy to see that
$\tau$ depends on $\calF$ as well as the underlying network $G(\calV, \calE)$, and it is finite.
\begin{definition}
A source component \footnote{The definition of a source is different from \cite{vaidya2012IABC}.} $S$ of a given graph $G(\calV, \calE)$ is the collection of agents each of which has a directed path to every other agent in $G(\calV, \calE)$.
\end{definition}
It can be easily checked that if the source component $S$, if any, is a strongly connected component in $G(\calV, \calE)$. In addition, a graph contains at most one source component.

%

\subsection{Necessary Condition}
%
%
We now present a necessary condition on the underlying communication graph $G(\calV, \calE)$ for solving Problem \ref{goal}. Our necessary condition is based on characterizing the connectivity of each reduced graph of $G(\calV, \calE)$.
\begin{theorem}
\label{scsize}
Given a graph $G(\calV, \calE)$, if there exists a correct algorithm that can solve Problem \ref{goal} when the agents do not have knowledge of the matrix, under any assignment matrix $\bf A$ for any $k$ solution-redundant input functions, then a source component must exist containing at least $\max\{f+1, sp({\bf A})\}$ nodes.
\end{theorem}
The proof of Theorem \ref{scsize} can be found in Appendix \ref{app:SRNC}. 

For future reference, we term the necessary condition in Theorem \ref{scsize} as Condition 1. 
%
%
Condition 1 also implies a lower bound on the number of agents needed, stated below.
\begin{corollary}
\label{gsize}
For a given graph $G(\calV, \calE)$, if Condition 1 is true, then $n\ge \max \{sp({\bf A})+2f, 3f+1\}$.
\end{corollary}

It can be shown that this lower bound is indeed tight. For instance, the complete graph of size $sp({\bf A})+2f$, denoted by $K_{sp({\bf A})+2f}$. It can be easily proved by contradiction that $K_ {sp({\bf A})+2f}$ satisfies Condition 1. The proof of Corollary \ref{gsize} is presented in Appendix \ref{app:SRNC}.


\subsection{Sufficiency of Condition 1}
Let $\{\alpha(t)\}_{t=0}^{\infty}$ be a sequence of stepsizes such that $\alpha(t)\le \alpha(t+1)$ for all $t\ge 0$, $\sum_{t=0}^{\infty} \alpha(t)=\infty$, and $\sum_{t=0}^{\infty} \alpha^2(t)<\infty$.
We show that Condition 1 is also sufficient. 
Let $\phi=\left | \calF\right |$. Thus $\phi\le f$. Without loss of generality, let us assume that the non-faulty agents are indexed as 1 to $n-\phi$.
Recall that the system is synchronous.
If a non-faulty agent does not receive an expected message from an incoming neighbor (in the {\em Receive step} below), then that message is assumed to have some default value. With the exception of the update step (\ref{e_Z}) below, the algorithm is similar to the consensus algorithms in \cite{vaidya2012iterative,Vaidya2012MatrixConsensus,Nedic2009}.
\paragraph{}
\vspace*{8pt}\hrule

{\bf Algorithm 2}
\vspace*{4pt}\hrule

~

Steps to be performed by agent $i\in \calV-\calF$ in the $t$-th iteration:
\begin{enumerate}

\item {\em Transmit step:} Transmit current state $x_i(t-1)$ on all outgoing edges.
\item {\em Receive step:} Receive values on all incoming edges. These values form
multiset\footnote{In a multiset, multiple instances of of an element is allowed.  For instance, $\{1, 1, 2\}$ is a multiset. } $r_i(t)$ of size $|N_i^{-}|$.


\item {\em Update step:}
Sort the values in $r_i(t)$ in an increasing order, and eliminate
the smallest $f$ values, and the largest $f$ values (breaking ties
arbitrarily).
 Let $N_i^*(t)$ denote the identifiers of agents from
whom the remaining $|N_i^{-}| - 2f$ values were received, and let
$w_j$ denote the value received from agent $j\in N_i^*(t)$.
For convenience, define $w_i=x_i(t-1)$. \footnote{Observe that
if $j\in \{i\}\cup N_i^*(t)$ is non-faulty, then $w_j=x_j(t-1)$.}

Update its state as follows.
\begin{eqnarray}
x_i(t) ~ = ~\sum_{j\in \{i\}\cup N_i^*(t)} a_i \, w_j-\alpha(t-1)~ d_i(t-1),
\label{e_Z}
\end{eqnarray}
where $ a_i ~=~ \frac{1}{|N_i^*(t)|+1}$ and $d_i(t-1)$ is a gradient of agent $i$'s objective function $g_i(\cdot)$ at $x=x_i(t-1)$.
\end{enumerate}

~
\hrule

~

~

Recall that
 $i\not\in N_i^*(t)$
because $(i,i)\not\in\calE$.
The ``weight'' of each term on the right-hand side of
(\ref{e_Z}) is $a_i$, and these weights add to 1. Observe that $0<a_i\leq 1$.
Let ${\bf x}\in \reals^{n\times \phi}$, be a real vector of dimension $n-\phi$, with $x_i$ being the local estimate of agent $i, \forall\,  i\in \calV-\calF$. Thus, ${\bf x}(t)$ is a vector of the local estimates of non-faulty agents at iteration $t$.

Since $G(\calV,\calE)$ satisfies Condition 1, as shown in \cite{Vaidya2012MatrixConsensus},  the updates of ${\bf x}\in \reals^{n-\phi}$ in each iteration can be written compactly in a matrix form.
%
%
\begin{align}
\label{update}
{\bf x}(t+1)={\bf M}(t){\bf x}(t)- \alpha(t) {\bf d}(t).
\end{align}
The construction of ${\bf M}(t)$ and relevant properties are given in \cite{Vaidya2012MatrixConsensus} and are also presented in Appendix \ref{app:matrix update} for completeness. Let $\calH\in R_{\calF}$ be a reduced graph of the given graph $G(\calV, \calE)$ with ${\bf H}$ as adjacency matrix. It is shown that in every iteration $t$, and for every ${\bf M}(t)$, there exists a reduced graph $\calH (t)\in R_{\calF}$ with adjacency matrix ${\bf H}(t)$ such that
\begin{align}
\label{reducedgraph}
{\bf M}(t)\ge \beta{\bf H}(t),
\end{align}
where $0<\beta<1$ is a constant. The definition of $\beta$ can be found in \cite{Vaidya2012MatrixConsensus}. Equation (\ref{update}) can be further expanded out as
\begin{align}
{\bf x}(t+1)&={\bf \Phi} (t, 0){\bf x}(0)-\sum_{r=1}^{t+1}\alpha(r-1){\bf \Phi} (t, r){\bf d}(r-1),
\label{updates}
\end{align}
where
${\bf \Phi}(t,r)={\bf M}(t){\bf M}(t-1)\ldots {\bf M}(r)$ and by convention ${\bf \Phi}(t,t)={\bf M}(t)$ and ${\bf \Phi} (t, t+1)={\bf I}_{n-\phi}$, the identity matrix. Note that ${\bf \Phi}(t,r)$ is a backward product (i.e., therein index decrease from left to right in the product).

\subsubsection{Convergence of the Transition Matrices ${\bf \Phi}(t,r)$ }
\label{ConvergenceProduct}

It can be seen from (\ref{updates}) that the evolution of estimates of non-faulty agents ${\bf x}(t)$ is determined by the backward product ${\bf \Phi}(t,r)$. Thus, we first characterize the evolutional properties and limiting behaviors of the backward product ${\bf \Phi}(t,r)$, assuming that the given $G(\calV, \calE)$ satisfies Condition 1.

Let $k^{\prime}=sp({\bf A})$.
The following lemma describes the structural property of ${\bf \Phi}(t,r)$ for sufficient large $t$. For a given $r$, Lemma \ref{lb} states that all non-faulty agents will be influenced by at least $\max \{k^{\prime}, f+1\}$ common non-faulty agents, and this set of influencing agents may depend on $r$. Proof of Lemma \ref{lb} can be found in Appendix \ref{app:ConvergenceProduct}.

\begin{lemma}
\label{lb}
There are at least $\max \{sp({\bf A}), f+1\}$ columns in ${\bf \Phi}(r+\nu-1,r)$ that are lower bounded by $\beta^{\nu}\ones$ component-wise for all $r$, where $\ones\in \reals^{n-\phi}$ is an all one column vector of dimension $n-\phi$.
\end{lemma}
Using coefficients of ergodicity theorem, it is showed in \cite{Vaidya2012MatrixConsensus} that if the given graph $G(\calV, \calE)$ satisfies Condition 1, then ${\bf \Phi}(t,r)$ is weak-ergodic. Moreover, because weak-ergodicity is equivalent to strong-ergodicity for backward product of stochastic matrices \cite{1977Seneta}, as $t\diverge$ the limit of ${\bf \Phi}(t,r)$ exists
\begin{align}
\label{mixing}
\lim_{t\ge r,~ t\diverge}{\bf \Phi}(t, r)=\ones {\bf \pi}(r),
\end{align}
where ${\bf \pi}(r)\in \reals^{n-\phi}$ is a row stochastic vector (may depend on $r$).  It is shown, using ergodic coefficients, in \cite{Anthonisse1977360} that the rate of the convergence in (\ref{mixing}) is exponential, as formally stated in Theorem \ref{convergencerate}. Recall that $\tau=|R_{\calF}|$, $n-\phi$ is the total number of non-faulty agents, and $0<\beta<1$ is a constant for which (\ref{reducedgraph}) holds.

\begin{theorem}\cite{Anthonisse1977360}
\label{convergencerate}
Let $\nu=\tau(n-\phi)$ and $\gamma=1-\beta^{\nu}$. For any sequence ${\bf \Phi}(t, r)$,
\begin{align}
\left | {\bf \Phi}_{ij}(t, r)-\pi_j(r)\right |\le \gamma^{\lceil\frac{t-r+1}{\nu}\rceil},
\end{align}
for all $t\ge r$.
\end{theorem}

Our next lemma is an immediate consequence of Lemma \ref{lb} and the convergence of ${\bf \Phi}(t, r)$, stated in (\ref{mixing}).
\begin{lemma}
\label{lblimiting}
For any fixed $r$, there exists a subset $\calI_r\subseteq \calV-\calF$ such that $|\calI_r|\ge \max \{sp({\bf A}), f+1\}$ and for each $i\in \calI_r$,
\begin{align*}
\pi_i(r)\ge \beta^{\nu}.
\end{align*}
\end{lemma}
The proof of Lemma \ref{lblimiting} can be found in Appendix \ref{app:ConvergenceProduct}.

\subsubsection{Convergence Analysis of Algorithm 2}
Here, we study the convergence behavior of Algorithm 2. The structure of our convergence proof is rather standard, which is also adopted in \cite{Duchi2012,Nedic2009,ram2010distributed,Tsianos2012,tsitsiklis1986distributed}. We have shown that the evolution dynamics of ${\bf x}(t)$ is captured by (\ref{update}) and (\ref{updates}).
Suppose that all agents, both non-faulty agents and faulty agents cease computing $d_i(t)$ after some time $\bar{t}$, i.e., after $\bar{t}$ subgradient is replaced by 0.

Let $\{\bar{\bf x}(t)\}$ be the sequences of local estimates generated by the non-faulty agents in this case. From (\ref{updates}) we get
\begin{align*}
\bar{\bf x}(t)={\bf x}(t),
\end{align*}
for all $t\le \bar{t}$. From (\ref{update}) and (\ref{updates}), we have for all $s\ge 0$, it holds that
\begin{align}
\bar{\bf x}(\bar{t}+s+1)&={\bf \Phi} (t, 0){\bf x}(0)-\sum_{r=1}^{\bar{t}}\alpha(r-1){\bf \Phi} (\bar{t}+s, r){\bf d}(r-1).
\label{evo}
\end{align}
Note that the summation in RHS of (\ref{evo}) is over $\bar{t}$ terms since all agents cease computing $d_j(t)$ starting from iteration $\bar{t}$. As $s\diverge$, we have

\begin{align}
\lim_{s\diverge}
\bar{\bf x}(\bar{t}+s+1)&=\lim_{s\diverge}{\bf \Phi} (t, 0){\bf x}(0)-\sum_{r=1}^{\bar{t}}\alpha(r-1){\bf \Phi} (\bar{t}+s, r){\bf d}(r-1)\nonumber\\
&=\lim_{s\diverge}{\bf \Phi} (t, 0){\bf x}(0)-\pth{\sum_{r=1}^{\bar{t}}\alpha(r-1)\lim_{s\diverge}{\bf \Phi} (\bar{t}+s, r){\bf d}(r-1)}\nonumber\\
&=\ones {\bf \pi}(0){\bf x}(0)-\pth{\sum_{r=1}^{\bar{t}}\alpha(r-1)\ones {\bf \pi}(r){\bf d}(r-1)}\nonumber\\
&= \pth{\iprod{\pi(0)}{{\bf x}(0)}-\sum_{r=1}^{\bar{t}}\alpha(r-1) \iprod{\pi(r)}{{\bf d}(r-1)}} \ones,
\label{identical}
\end{align}
where $\iprod{\cdot}{\cdot}$ is used to denote the inner product of two vectors of proper dimension. Let ${\bf y}(\bar{t})$ denote the limiting vector of $\bar{\bf x}(\bar{t}+s+1)$ as $s+1\diverge$. Since all entries in the limiting vector are identical we denote the identical value by $y(\bar{t})$. Thus, ${\bf y}(\bar{t})=[y(\bar{t}), \ldots, y(\bar{t})]^{\prime}$.

From (\ref{identical}) we have
\begin{align}
y(\bar{t})=\iprod{\pi(0)}{{\bf x}(0)}-\sum_{r=1}^{\bar{t}}\alpha(r-1) \iprod{\pi(r)}{{\bf d}(r-1)}.
\label{yupdate}
\end{align}
If, instead, all agents cease computing $d_i(t)$ after iteration $\bar{t}+1$,  then the identical value, denoted by $y(\bar{t}+1)$, equals
\begin{align}
\nonumber
y(\bar{t}+1)&=\iprod{\pi(0)}{{\bf x}(0)}-\sum_{r=1}^{\bar{t}+1}\alpha(r-1) \iprod{\pi(r)}{{\bf d}(r-1)}\\
\nonumber
&=\iprod{\pi(0)}{{\bf x}(0)}-\sum_{r=1}^{\bar{t}}\alpha(r-1) \iprod{\pi(r)}{{\bf d}(r-1)}-\alpha(\bar{t})  \iprod{\pi(\bar{t}+1)}{{\bf d}(\bar{t})}\\
&=y(\bar{t})-\alpha(\bar{t})  \iprod{\pi(\bar{t}+1)}{{\bf d}(\bar{t})},
\label{ydynamic}
\end{align}
where each entry $d_i(\bar{t})$ in ${\bf d}(\bar{t})$ denotes the subgradient of $g_i(\cdot)$ computed by agent $i$ at $x_i(\bar{t})$. With a little abuse of notation, henceforth we use $t$ to replace $\bar{t}$. The actual reference of $t$ should be clear from the context.\\

In our convergence analysis, we will use the well-know ``almost supermartingale" convergence theorem in \cite{Robbins1985}, which can also be found as Lemma 11, in Chapter 2.2 \cite{poljak1987introduction}. We present a simpler deterministic version of the theorem in the next lemma.

\begin{lemma}\cite{Robbins1985}
\label{stochatic convergence}
Let $\{ a_t\}_{t=0}^{\infty}, \{ b_t\}_{t=0}^{\infty}$, and $\{ c_t\}_{t=0}^{\infty}$ be non-negative sequences. Suppose that
$$a_{t+1}\le a_{t}-b_{t}+c_{t}~~~~~\text{for all}~t\ge 0,$$
and $\sum_{t=0}^{\infty} c_t<\infty$. Then $\sum_{t=0}^{\infty} b_t<\infty$ and the sequence $\{a_t\}_{t=0}^{\infty}$ converges to a non-negative value.
\end{lemma}

The basic iterative relation of the consensus value $y(t)$ is stated in our Lemma \ref{BasicIter}.

\begin{lemma}
\label{BasicIter}
Let $\{ y(t)\}_{t=0}^{\infty}$ be the sequence of limiting consensus value defined by (\ref{yupdate}), and $\{x_i(t)\}_{t=0}^{\infty}$ be the sequence for $i\in \calV-\calF$
generated by (\ref{updates}). Let $\{\delta_i(t)\}_{t=0}^{\infty}$ be a sequence of subgradients of $g_i$ at $y(t)$ for all $i\in \calV-\calF$. Then the following basic relations hold. For any $x\in \reals$ and any $t\ge 0$,
\begin{align*}
\left | y(t+1)-x\right |^2&\le \left |y(t)-x\right |^2+4L\alpha(t) \sum_{j=1}^{n-\phi} \pi_j(t+1)\left |y(t)-x_j(t)\right |\\
&\quad-2\alpha(t) \sum_{j=1}^{n-\phi} \pi_j(t+1)\pth{g_j\pth{y(t)}-g_j(x)}+\alpha^2(t)(n-\phi)L^2
\end{align*}
\end{lemma}

The proof of Lemma \ref{BasicIter} can be found in \cite{Nedic2009}. We present the proof in Appendix \ref{app:ConvergenceAlgorithm}.
%
%
%
%
 For each $t$ and each $i\in \calV-\calF$, the distance between the consensus value $y(t)$ and the local estimate $x_i(t)$ is bounded from above.
\begin{lemma}
\label{uub}
Let $U=\max_{i\in \calV-\calF} x_i(0)$, and $u=\min_{i\in \calV-\calF} x_i(0)$. For every $i\in \calV-\calF$, a uniform bound on $|y(t)-x_i(t)|$ for $t\ge 1$ is given by:
\begin{align}
\label{consensus ub}
\left |y(t)-x_i(t)\right |\le \pth{n-\phi}\max \{|u|, |U|\}\gamma^{\lceil \frac{t}{\nu}\rceil} +\pth{n-\phi}L\sum_{r=1}^{t-1}\alpha(r-1) \gamma^{\lceil \frac{t-r}{\nu}\rceil}+2\alpha(t-1) L.
\end{align}
When $t=1$, $\sum_{r=1}^{t-1}\alpha(r-1) \gamma^{\lceil \frac{t-r}{\nu}\rceil}=0$ by convention.
\end{lemma}
Note that the upper bound on $|y(t)-x_i(t)|$ in (\ref{consensus ub}) depends on $t$. In fact, this upper bound will diminish over time, as formally stated below.

\begin{lemma}
\label{consensus}
For each $i\in \calV-\calF$, the limit of $|y(t)-x_i(t)|$ exists and
$$\lim_{t\diverge} |y(t)-x_i(t)|=0.$$
\end{lemma}

Our main convergence result is stated below.
\begin{theorem}[Convergence]
For each $i\in \calV-\calF$, $\{x_i(t)\}_{t=0}^{\infty}$ converges to the same optimum in $X$, i.e.,
$$\lim_{t\diverge}|x_i(t)-x^*|=0,$$ where $x^*\in X$.
\label{ConAlgorithm1}
\end{theorem}
We provide a sketch of the convergence proof below. Formal proof can be found in Appendix \ref{app:ConvergenceAlgorithm}.\\

Recall that each $g_i(\cdot)$ is defined as
\begin{align*}
g_i(x)={\bf A}_{1i}h_1(x)+{\bf A}_{2i}h_2(x)+\ldots +{\bf A}_{ki}h_k(x),
\end{align*}
for $i\in \calV$, where ${\bf A}_{ji}\ge 0$ and $\sum_{j=1}^k{\bf A}_{ji}=1$. Let $Y^i=\argmin~ g_i(x)$ and $Y_j^i=\argmin~{\bf A}_{ji}h_j(x)$ for $j=1, \ldots, k$. Since for each $j\in \{1, \ldots, k\}$ such that ${\bf A}_{ji}=0$, $\argmin~{\bf A}_{ji}h_j(x)=0$ is a constant function over the whole real line, it holds that $Y_j^i=\reals$. Since positive constant scaling does not affect the optimal set of a function, for each $j\in \{1, \ldots, k\}$ such that ${\bf A}_{ji}>0$, it holds that $Y_j^i=X_j$.  In addition, because $h_1(x), \ldots, h_k(x)$ are solution redundant functions, i.e., $\cap_{j=1}^k X_j\not=\O$,  functions ${\bf A}_{1i}h_1(x), \ldots, {\bf A}_{ki}h_k(x)$ are also solution redundant.  It can be shown (formally proved in Appendix \ref{app: Connection}) that
$$Y^i=\cap_{j: {\bf A}_{ji}>0}X_j\supseteq \cap_{j=1}^k X_j=X, ~ \text{for all}~ i\in \calV.$$
Let $x^{\prime}\in X$. Define $g_j^*$ as the optimal value of function $g_j(\cdot)$ for each $j\in \calV$. We have
\begin{align}
\label{almost monotone 1}
\nonumber
\left |y(t+1)-x^{\prime}\right |^2&\le \left |y(t)-x^{\prime}\right |^2+4L\alpha(t) \sum_{j=1}^{n-\phi} \pi_j(t+1)\left |y(t)-x_j(t)\right |\\
\nonumber
&\quad-2\alpha(t) \sum_{j=1}^{n-\phi} \pi_j(t+1)\pth{g_j\pth{y(t)}-g_j(x^{\prime})}+\alpha^2(t)(n-\phi)L^2\\
\nonumber
&\overset{(a)}{=}\left |y(t)-x^{\prime}\right |^2+4L\alpha(t) \sum_{j=1}^{n-\phi} \pi_j(t+1)\left |y(t)-x_j(t)\right |\\
&\quad-2\alpha(t) \sum_{j=1}^{n-\phi} \pi_j(t+1)\pth{g_j\pth{y(t)}-g_j^*}+\alpha^2(t)(n-\phi)L^2.
\end{align}
Equality $(a)$ holds because of $x^{\prime}\in X\subseteq Y^j$ for each $j\in \calV$, then $g_j(x^{\prime})=g_j^*$.\\

For each $t\ge 0$, define
\begin{align*}
a_t&=|y(t)-x^{\prime}|^2,\\
b_t&= 2\alpha(t) \sum_{j=1}^{n-\phi} \pi_j(t+1)\pth{g_j(y(t))-g_j^*},\\
c_t&=4L\alpha(t) \sum_{j=1}^{n-\phi} \pi_j(t+1)|y(t)-x_j(t)|+\alpha^2(t)(n-\phi)L^2.
\end{align*}
It is easy to see that $a_t\ge 0$ and $c_t\ge 0$ for each $t$.
Since $g_j^*$ is the optimal value of function $g_j(\cdot)$, it holds that $b_t\ge 0$ for each $t$. Thus, $\{a_t\}_{t=0}^{\infty}, \{b_t\}_{t=0}^{\infty}$ and $\{c_t\}_{t=0}^{\infty}$ are three non-negative sequences. By (\ref{almost monotone 1}), it holds that
$$a_{t+1}\le a_t-b_t+c_t ~~~~\text{for each}~t\ge 0.$$
To apply Lemma \ref{stochatic convergence}, we need to show that $\sum_{t=0}^{\infty} c_t<\infty$.
In fact, the following lemma holds.
\begin{lemma}
\label{c1}
\begin{align*}
\sum_{t=0}^{\infty}\alpha(t) \sum_{j=1}^{n-\phi} \pi_j(t+1)|y(t)-x_j(t)|<\infty.
\end{align*}
\end{lemma}
The proof of Lemma \ref{c1} is presented in Appendix \ref{app:ConvergenceAlgorithm}. In addition, since $\sum_{t=0}^{\infty}\alpha^2(t)<\infty$, it holds that
$$(n-\phi)L^2\sum_{t=0}^{\infty}\alpha^2(t)<\infty.$$
Thus, we get
\begin{align*}
\nonumber
\sum_{t=0}^{\infty} c_t&=\sum_{t=0}^{\infty}\pth{4L\alpha(t) \sum_{j=1}^{n-\phi} \pi_j(t+1)|y(t)-x_j(t)|+\alpha^2(t)(n-\phi)L^2}\\
\nonumber
&=4L\sum_{t=0}^{\infty}\pth{\alpha(t) \sum_{j=1}^{n-\phi} \pi_j(t+1)|y(t)-x_j(t)|}+(n-\phi)L^2\sum_{t=0}^{\infty}\alpha^2(t)\\
&< \infty.
\end{align*}
Therefore, applying Lemma \ref{stochatic convergence} to the sequences $\{a_t\}_{t=0}^{\infty}, \{b_t\}_{t=0}^{\infty}$ and $\{c_t\}_{t=0}^{\infty}$, we have that for any $x^{\prime}\in X$,
$a_t=|y(t)-x^{\prime}|$ converges, and
\begin{align*}
\sum_{t=0}^{\infty}b_t=\sum_{t=0}^{\infty}\alpha(t) \sum_{j=1}^{n-\phi}\pi_j(t+1)\pth{g_j(y(t))-g_j^*}<\infty.
\end{align*}
Since $|y(t)-x^{\prime}|$ converges for any fixed $x^{\prime}\in X$, by definition of sequence convergence and the dynamic of $y(t)$ in (\ref{yupdate}), it is easy to see that $y(t)$ also converges. Let $\lim_{t\diverge}y(t)=y$. Next we show that $y\in X$.\\

%
%
%
%
%
Let $\calI_{t+1}\subseteq \calV-\calF$ be the set of indices such that for each $j\in \calI_{t+1}$, $\pi_j(t+1)\ge \beta^{\nu}$. As $G(\calV, \calE)$ satisfies Condition 1,  $|\calI_{t+1}|\ge \max \{k^{\prime}, f+1\}$.  Since $g_j(y(t))-g_j^*\ge 0$ for all $j$,
then
\begin{align*}
\nonumber
\sum_{j=1}^{n-\phi} \pi_j(t+1)\pth{g_j(y(t))-g_j^*}&\ge \sum_{j\in \calI_{t+1}} \pi_j(t+1)\pth{g_j(y(t))-g_j^*}\\
\nonumber
&\ge \beta^{\nu}\sum_{j\in \calI_{t+1}}\pth{g_j(y(t))-g_j^*}\\
\nonumber
&=\beta^{\nu}\sum_{j\in \calI_{t+1}}\sum_{i=1}^{k}{\bf A}_{ij}\pth{h_i(y(t))-h_i^*}\\
\nonumber
&=\beta^{\nu}\sum_{i=1}^{k}\pth{\sum_{j\in \calI_{t+1}}{\bf A}_{ij}}\pth{h_i(y(t))-h_i^*}\\
&\ge k\beta^{\nu}C_2 \pth{h(y(t))-h^*},
\end{align*}
where
\begin{align*}
C_2=\min_{\calI\subseteq \calV: \,  |\calI|\ge \max \{k^{\prime}, f+1\}} \sum_{i\in \calI} {\bf A}_{ij},
\end{align*}
and the last inequality follows from the fact that $h_i(y(t))-h_i^*\ge 0$. In addition, as $sp({\bf A})=k^{\prime}$, then $\sum_{i\in \calI} {\bf A}_{ij}>0$ for every $\calI\subseteq \calV: \,  |\calI|\ge \max \{k^{\prime}, f+1\}$. Since $\bf A$ is finite, $C_2$ is well-defined and $C_2>0$.
%
If $y\notin X$, it can be shown that $k\beta^{\nu}C_2 \pth{h(y(t))-h^*}=\infty$. This contradicts the fact that $\sum_{t=0}^{\infty}b_t<\infty$. Thus, $y\in X$. \\

Therefore, we conclude that limit of $|x_i(t)-y|$ exists and $$\lim_{t\diverge} |x_i(t)-y|=0,$$
proving Theorem \ref{ConAlgorithm1}.

\section{Summary and Conclusion}\label{sec:summary}

In this report, we introduce the condition-based approach to Byzantine multi-agent optimization. We have shown that when there is enough redundancy in the local cost functions, or in the local optima, Problem \ref{goal} can be solved iteratively.

Two slightly different variants are considered: condition-based Byzantine multi-agent optimization with side information and condition-based Byzantine multi-agent optimization without side information. For the former, when side information is available at each agent, a decoding-based algorithm is proposed, assuming each input function is differentiable. This algorithm combines the gradient method with the decoding procedure introduced in \cite{candes2005decoding} by choosing  proper ``generator matrices" as job assignment matrices. With such a decoding subroutine, our algorithm essentially performs the gradient method, where gradient computation is processed distributedly over the multi-agent system.
When side information is not available at each agent, we propose a simple consensus-based algorithm in which each agent carries minimal state across iterations. This consensus-based algorithm solves Problem \ref{goal} under the additional assumption over input functions that all input functions share at least one common optimum. Although the consensus-based algorithm can only solve Problem \ref{goal} for a restricted class of input functions, nevertheless, as each non-faulty agent does not need to store the job matrix ${\bf A}$ throughout execution and does not need to perform the decoding procedure at each iteration, the requirements on memory and computation are less stringent comparing to the decoding-based algorithm.  In addition, in contrast to the decoding-based algorithm, the consensus-based algorithm also works for nonsmooth input functions. Thus, the consensus-based algorithm may be more practical in some applications.

\bibliographystyle{abbrv}
\bibliography{PSDA_DL}

\newpage
\appendix

\setlength {\parskip}{6pt}

\centerline{\Large\bf Appendices}

\section{Connection between $X$ and $X_j$'s for $j=1, \ldots, k$}
\label{app: Connection}

Recall that $X_j$ is the set of optimal solution(s) of input function $h_j(x)$, for $j=1, \ldots, k$; and that $X$ is the optimal set of function $\frac{1}{k}\sum_{j=1}^k h_j(x)$ in Problem \ref{goal}. Propositions \ref{p1} and \ref{p2} are used in proving the correctness of other results in this report.

\begin{proposition}\cite{SuV15MultiAgentPartI}
\label{p1}
The optimal set $X$ of Problem \ref{goal} is contained in the convex hull of the union of all $X_j$'s, i.e.,
\begin{align}
X\subseteq Cov\pth{\cup_{j=1}^k X_j},
\end{align}
where $Cov\pth{Z}$ is the convex hull of set $Z$.
\end{proposition}

The above proposition holds for any collection of $k$ admissible input functions. A stronger connection holds for solution-redundant input functions, as stated below.

\begin{proposition}
\label{p2}
When all the input functions share at least one common optimum, i.e., $\cap_{j=1}^{k}X_j\not=\O$, then
\begin{align}
X=\cap_{j=1}^{k}X_j.
\end{align}
\end{proposition}
\begin{proof}

We first show that $\cap_{j=1}^{k}X_j\subseteq X$. Let $h_j^*$ be the optimal value of function $h_j(x)$ for $j=1, \ldots, k$, and let $h^*$ be the optimal value of function $\frac{1}{k}\sum_{j=1}^k h_j(x)$. Since $\cap_{j=1}^{k}X_j\not=\O$, let $x_0\in \cap_{j=1}^{k}X_j$.
Then for all $x\in \reals$
\begin{align*}
\frac{1}{k}\sum_{j=1}^k h_j(x_0)=\frac{1}{k}\sum_{j=1}^kh_j^*\le h^*\le \frac{1}{k}\sum_{j=1}^k h_j(x).
\end{align*}
So we know
\begin{align*}
\sum_{j=1}^k h_j(x_0)= \sum_{j=1}^k h_j^*=h^*.
\end{align*}
Thus  $x_0\in X$ and $\cap_{j=1}^{k}X_j\subseteq X$. 

Next we show $X\subseteq \cap_{j=1}^{k}X_j$. We prove this by contradiction. Suppose on the contrary that $X\not\subseteq \cap_{j=1}^{k}X_j$, then there exists $x^{\prime}\in X$ such that $x^{\prime}\not\in \cap_{j=1}^{k}X_j$. The latter implies that $x^{\prime}\not\in X_{j_0}$ for some $j_0\in \{1, \ldots, k\}$. Then
\begin{align*}
\sum_{j=1}^k h_j(x^{\prime})&=\pth{\sum_{1\le j\le k, j\not=j_0} h_j(x^{\prime})}+h_{j_0}(x^{\prime})\\
&\ge \pth{\sum_{1\le j\le k, j\not=j_0} h_j^*}+h_{j_0}(x^{\prime})\\
&> \pth{\sum_{1\le j\le k, j\not=j_0} h_j^*}+h_{j_0}^*\\
&=\sum_{j=1}^k h_j^*=h^*.
\end{align*}
So $x^{\prime}\notin X$, which leads to a contradiction. Thus $X\subseteq \cap_{j=1}^kX_j$.

Therefore, $X=\cap_{j=1}^kX_j$.

\raggedleft $\square$
\end{proof}

\section{Condition-based Byzantine multi-agent optimization without side information}
\label{app:SRNC}
%
\subsection*{Proof of Lemma \ref{sparsity}}
\begin{proof}
Let $sp({\bf A})=k^{\prime}$, by definition of $sp\pth{\bf A}$, the sum vector of any collection of $k^{\prime}$ columns of ${\bf A}$ is component-wise positive. Suppose there exists a row, say $i_0$, that contains at least $k^{\prime}$ zero entries. Let $j_1, j_2, \ldots, j_{k^{\prime}}$ be any $k^{\prime}$ columns in $\bf A$, wherein the $i_0$--coordinate of each column is zero.
Thus the $i_0$--th coordinate of $\sum_{r=1}^{k^{\prime}} {\bf A}_{j_r}$ is zero, contradicting the hypothesis that $sp({\bf A})=k^{\prime}$. In addition, if every row contains more than $k^{\prime}-1$ zeros, using the same argument it can be shown that $k^{\prime}$ is not the smallest integer.

Conversely, we need to show that if there are at most $k^{\prime}-1$ zero entries in each row of $\bf A$ and there exists one row contains exactly $k^{\prime}-1$ zero entries, then $sp({\bf A})=k^{\prime}$. Let $j_1, \ldots, j_{k^{\prime}}$ be any collection of $k^{\prime}$ columns of $\bf A$. For each coordinate, at least one of the chosen $k^{\prime}$ columns contains positive entry in that coordinate. So we have $\sum_{r=1}^{k^{\prime}} {\bf A}_{j_r}>\zeros$ componentwise. By definition of $sp({\bf A})$, we know $sp({\bf A})\le k^{\prime}$. In addition, let $i_0$ be a row in which there are exactly $k^{\prime}-1$ zeros, then there exists a collection of $k^{\prime}-1$ columns whose $i_0$--th coordinate are all zeros, and that the sum of the $k^{\prime}-1$ columns also has the $i_0$--th coordinate being 0. Thus $sp({\bf A})=k^{\prime}$.

\raggedleft $\square$
\end{proof}

\subsection*{Proof of Theorem \ref{scsize}}
\begin{proof}

We first show that if Problem \ref{goal} is solvable, then a source component must exist in every reduced graph of $G(\calV,\calE)$, containing at least $f+1$ nodes. Then we show that when $k^{\prime}>f+1$, the source component must contain at least $k^{\prime}$ nodes. These two claims together show that if Problem \ref{goal} is solvable, then a source component must exist in every reduced graph of $G(\calV,\calE)$, containing at least $\max \{f+1, k^{\prime}\}$ nodes, proving the theorem.

\begin{definition}[Condition 2]
Given a graph $G(\calV, \calE)$, for any node partition $L, R, C, F$ of $G(\calV, \calE)$ such that $L, R$ are nonempty and $|F|\le f$, one of the following must hold:
(1) there exists a node $i\in L$ that has at least $f+1$ incoming neighbors in $R\cup C$, i.e., $|N_i^{-}\cap (R\cup C)|\ge f+1$; or
(2) there exists a node $j\in R$ that has at least $f+1$ incoming neighbors in $L\cup C$, i.e., $|N_j^{-}\cap (L\cup C)|\ge f+1$.
\end{definition}
Now we show that Condition $2$ is a necessary condition
for Problem 1. Suppose graph $G(\calV, \calE)$ does not satisfy Condition 2 and there exists a correct algorithm $\calA$ solving Problem 1 for solution redundant input functions. Recall that $X_i=\argmin~h_i(x)$ for each $i=1, \ldots, k$, and $X=\argmin~h(x)$. Consider the input functions $h_i(x)$ for all $i=1, \ldots, k$ such that each $h_i(x)$ is admissible with optimal set $X_i=[0,1]$. Since $\cap_{i=1}^k X_i=[0,1]\not=\O$, then $h_1(x), \ldots, h_k(x)$ is a collection of $k$ solution redundant input functions.
In addition, by Proposition \ref{p2}, we know that $X=[0,1]$. Since $G(\calV, \calE)$ does not satisfy Condition 2, then there exists a node partition $L, R, C, F$, where $L, R$ are nonempty and $|F|\le f$, such that $|N_i^{-}\cap (R\cup C)|\le f$ for each $i\in L$ and $|N_j^{-}\cap (L\cup C)|\le f$ for each $j\in R$. Consider the execution, denoted by $e_1$, wherein all nodes in $F$ are faulty and all the remaining nodes are non-faulty. The initial states of all non-faulty nodes are assigned as follows: $x_i(0)=0$ for each $i\in L$, $x_i(0)=1$ for each $i\in R$, and $x_i(0)$ as an arbitrary value within $[0,1]$ for each $i\in C$. In iteration 1, each faulty node $p$ sends $F_p(0, g_p(\cdot))$ to nodes in $L$, sends $F_p(1, g_p(\cdot))$ to nodes in $R$ and sends $F_p(a, g_p(\cdot))$ to nodes in $C$, where $a$ is an arbitrary value within $[0,1]$.
Let $i\in L$ be an arbitrary node in $L$. We will show that there exists an execution $e_i$ that can not be distinguished from $e_1$ by node $i$. Thus node $i$ should behave in the same way in $e_1$ and $e_i$. Let $x_i(t)$ and $\bar{x}_i(t)$ be the local estimate of agent $i$ in $e_1$ and in $e_i$, respectively.

\noindent{{\bf Execution $e_i$:}} The input functions are $h_1(x), \ldots, h_k(x)$. All nodes in $N_i^{-}\cap (R\cup C)$ are faulty, and the other nodes are non-faulty with initial state 0, i.e., $\bar{x}_j(0)=0$ for all $j\notin  N_i^{-}\cap (R\cup C)$.\footnote{Execution $e_i$ is possible since $|N_i^{-}\cap (R\cup C)|\le f$. } Since $\bar{x}_i(0)=0\in [0,1]=X$, for all $i\in \calV-\calF$, then $\bar{x}_i(1)=0$ for all $i\in \calV-\calF$, where $\calF=N_i^{-}\cap (R\cup C)$ is the set of faulty nodes in execution $e_i$.

Since agent $i$ cannot distinguish execution $e_i$ from $e_1$, thus $x_i(1)=0$. As agent $i$ is an arbitrary agent in $L$, in execution $e_1$ it holds that $x_i(1)=0$ for all $i\in L$. Similarly, we can show that in execution $e_1$, $x_i(1)=1$ for all $i\in R$. Repeatedly applying the above argument, we can conclude that for any iteration $t$ in execution $e_1$, it follows that $x_i(t)=0$ for all $i\in L$ and $x_j(t)=1$ for all $j\in R$, contradicting the asymptotic consensus requirement of a correct algorithm for Problem \ref{goal}. Thus, Condition 2 is a necessary condition for Problem \ref{goal}. In addition, it was shown in \cite{vaidya2012IABC} that graph $G(\calV, \calE)$ satisfies Condition 2 if and only if a source component exists containing at least $f+1$ nodes in every reduced graph of $G(\calV, \calE)$.

Therefore, if Problem \ref{goal} is solvable, a source component containing at least $f+1$ nodes must exist in every reduced graph of $G(\calV, \calE)$. Next we show, by contradiction, that if Problem \ref{goal} is solvable and $k^{\prime}>f+1$, a source component must contain at least $k^{\prime}$ nodes.\\

Suppose there exists a reduced graph
$\calH$ of $G(\calV, \calE)$ whose source component contains at most $k^{\prime}-1$ agents, and there exists a correct algorithm $\calA$ that can solve Problem \ref{goal}. Denote the source component of the reduced graph $\calH$ by $S_{\calH}$. Let $L, R, C, F$ be the node partition of $G(\calV, \calE)$ where $L=S_{\calH}$, $C=\O$, $R=\calV-L-\calF$ and $F=\calF$. Note that it is possible that $R=\O$. 
Let $h_1(x),\ldots, h_k(x)$ and $\widetilde{h}_1(x),\ldots, \widetilde{h}_k(x)$ be two collections of $k$ admissible input functions such that (1) $h_i(\cdot)=\widetilde{h}_i(\cdot)$ for $i=1, \ldots, k-1$, $\argmin~\widetilde{h}_i(x)=[0,1]$ for all $i=1,\ldots, k$, and $\argmin~ h_{k}(x)=\{1\}$. Let ${\bf A}$ be an assignment matrix such that $sp({\bf A})=k^{\prime}$ and ${\bf A}_{ki}=0$ for each $i\in L$. Such a matrix exists, since $|L|\le k^{\prime}-1$. Informally speaking, with this assignment matrix, each agent $i$ in $L$ does not have any information about the $k$--th input function.
Consider the execution $E_1$, wherein the input functions are $h_1(x),\ldots, h_k(x)$, each agent $p$ in $F$ is faulty and all other agents are non-faulty. The initial states of non-faulty agents in execution $E_1$ are assigned as follows: $x_i(0)=0$ for all $i\in L$, and $x_i(0)=1$ for all $i\in R$ (if $R\not=\O$)--recalling that $C=\O$.  Each faulty agent $p\in F$ sends $F_p\pth{0, \sum_{i=1}^k {\bf A}_{ip}\widetilde{h}_i(x)}$ to nodes in $L$, and sends $F_p\pth{1, \sum_{i=1}^k {\bf A}_{ip}h_i(x)}$ to nodes in $R$ (if $R\not=\O$).

Let $i\in L$ be an arbitrary agent in $L$. Now consider the following execution, denoted by $E_i$, wherein the local estimate of each non-faulty node $j$ is denoted as $\bar{x}_j$. The input functions are $\widetilde{h}_1(x), \ldots, \widetilde{h}_k(x)$. All nodes in $N_i^{-}\cap (R\cup C)$ are faulty, and the other nodes are non-faulty with initial state 0, i.e., $\bar{x}_j(0)=0$ for all $j\notin  N_i^{-}\cap (R\cup C)$.
 Since $\bar{x}_i(0)=0\in [0,1]=X$, for all $i\in \calV-\calF$, then $\bar{x}_i(t)=0$ for all $i\in \calV-\calF$ and for all $t$, where $\calF=N_i^{-}\cap (R\cup C)$ is the set of faulty nodes in execution $E_i$.

 Node $i$ cannot distinguish $E_i$ from $E_1$, thus $x_i(1)=0$ in $E_1$. Since $i$ is an arbitrary node in $L$, thus $x_i(1)=0$ for all $i\in L$ in $E_1$. Repeatedly applying the above argument, we have $\lim_{t\diverge}x_i(t)=0$ for each $i\in L$ in $E_1$. However, we know that in $E_1$, the optimal set is $X=\{1\}$. Because in execution $E_1$, the correct output must be $1$, $\calA$ is not a correct algorithm. Thus, we know if Problem \ref{goal} is solvable and $k^{\prime}>f+1$, a source component must contain at least $k^{\prime}$ nodes.\\

Therefore, we conclude that if Problem \ref{goal} is solvable, a source component containing at least $\max\{f+1, k^{\prime}\}$ nodes must exist in every reduced graph of $G(\calV, \calE)$, proving Theorem \ref{scsize}.

\raggedleft $\square$
\end{proof}

\subsection*{Proof of Corollary \ref{gsize}}
\begin{proof}
Let $sp({\bf A})=k^{\prime}$. It was shown in \cite{vaidya2012iterative} that if Condition 2 is true, then $n\ge 3f+1$. It is enough to consider the case when $k^{\prime}>f+1$.

Suppose $3f+1\le n<k^{\prime}+2f$. Consider the node partition $L, R, F$ such that $|R|=|F|=f$, and $L=\calV-R-F$. Since $3f+1\le n<k^{\prime}+2f$, it holds that $f+1\le |L|\le k^{\prime}-1$. Suppose all nodes in $F$ are faulty. Consider the subgraph $\calH$ constructed from $G(\calV, \calE)$ by (1) removing all faulty nodes, i.e., all nodes in $F$, and (2) for each $i\in L$, removing all incoming links from $R$. The subgraph $\calH$ is a valid reduced graph since $|R|=f$. By Theorem \ref{scsize}, a source component exists in $\calH$. Let $S$ be the source component of $\calH$. By Theorem \ref{scsize}, it holds that $|S|\ge k^{\prime}$. Since each node $j\in R$ cannot reach nodes in $L$, by definition, each node $j\in R$ is not contained in a source component. Thus, $S\subseteq L$. Consequently, it holds that $|S|\le |L|\le k^{\prime}-1$, contradicting the fact that $|S|\ge k^{\prime}$. Thus, when $k^{\prime}>f+1$, it holds that $n\ge k^{\prime}+2f$.\\

Therefore, we conclude that $n\ge \max\{3f+1, k^{\prime}+2f\}$.

\raggedleft $\square$
\end{proof}

\section{Matrix Representation of Algorithm 2}
\label{app:matrix update}

If $G(\calV,\calE)$ satisfies Condition 1, it was shown in Proposition \ref{claim_1} that the updates of ${\bf x}\in \reals^{n-\phi}$ in each iteration can be written compactly in a matrix form.
This observation is made in \cite{Vaidya2012MatrixConsensus}, and we restate this result below for completeness.

\begin{proposition}\cite{Vaidya2012MatrixConsensus}
\label{claim_1}
{
We can express the iterative update of the state
of a non-faulty node $i$ $(1\leq i\leq n-\phi)$
performed in (\ref{e_Z}) using the matrix form in (\ref{e_matrix_i})
below,
where ${\bf M}_i(t)$ satisfies the following four conditions.
\begin{eqnarray}
x_i(t+1) & = & {\bf M}_i(t) ~ {\bf x}(t)-\alpha(t) d_i(t).
\label{e_matrix_i}
\end{eqnarray}
}
In addition to $t$, the row vector ${\bf M}_i(t)$
may depend on the state vector ${\bf x}(t-1)$ as well as the
behavior of the faulty
nodes in $\calF$. For simplicity, the notation ${\bf M}_i(t)$ does not
explicitly represent this dependence.
\begin{enumerate}
\item ${\bf M}_i(t)$ is a {\em stochastic} row vector of size $(n-\phi)$.
Thus,
${\bf M}_{ij}(t)\geq 0$, for $1\leq j\leq n-\phi$, and
\[
\sum_{1\leq j\leq n-\phi}~{\bf M}_{ij}(t) ~ = ~ 1
\]

\item ${\bf M}_{ii}(t)$ equals $a_i$ defined in Algorithm 1.
Recall that $a_i\geq \alpha$.

\item ${\bf M}_{ij}(t)$ is non-zero
{\bf only if}  $(j,i)\in\calE$ or $j=i$.
\item At least $|N_i^-\cap\,(\calV-\calF)| - f+1$ elements in ${\bf M}_i(t)$
are lower bounded by some constant $\beta>0$
($\beta$ is independent of $i$).
Note that $N_i^-\cap\,(\calV-\calF)$ is the set of non-faulty
incoming neighbors of node $i$.
\end{enumerate}
\end{proposition}
%
%
%

\section{Convergence of the Transition Matrices ${\bf \Phi}(t,r)$ }
\label{app:ConvergenceProduct}

%
\subsection*{Proof of Lemma \ref{lb}}
\begin{proof}
Recall that $R_\calF$ is the collection of all reduced graphs of the given graph $G(\calV, \calE)$. Let $\calH\in R_\calF$ be an arbitrary reduced graph with adjacency matrix $\bf H$. Let $k^{\prime}=sp({\bf A})$. From Theorem \ref{scsize} we know that there are at least $\max \{k^{\prime}, f+1\}$ nodes in the unique source component in $\calH$. Denote the source component in $\calH$ by $S_{\calH}$ and let $j_1, j_2, \ldots, j_p$, where
\begin{align*}
p\triangleq|S_{\calH}|\ge \max \{k^{\prime}, f+1\},
\end{align*}
be the $p$ nodes in $S_{\calH}$.  By definition, each $j_i$ has a directed path to all the other non-faulty nodes in $\calH$. Since the length of a path from $j_i$ to any other node in $\calH$ is at most $n-\phi-1$, then the $j_i$--th column of ${\bf H}^{n-\phi}$ will be non-zero for $i=1, 2, \ldots, p$. Since $p\ge \max \{k^{\prime}, f+1\}$, there are at least $\max \{k^{\prime}, f+1\}$ such columns in ${\bf H}^{n-\phi}$.

Recall that for any $t\ge 1$, there exists a graph $\calH(t)\in R_\calF$ such that $\beta {\bf H}(t)\le {\bf M}(t)$, thus we have
\begin{align*}
{\bf \Phi}(r+\nu-1,r)&={\bf M}(r+\nu-1){\bf M}(r+\nu-2)\ldots {\bf M}(r)\\
&\ge \beta^{\nu} \prod_{t=r}^{r+\nu-1} {\bf H}(t).
\end{align*}
The above product of adjacency matrices consists of $\nu=\tau(n-\phi)$ matrices (corresponding to reduced graphs) in $R_\calF$. Thus, at least one of the $\tau$ distinct adjacency matrices in $R_\calF$, say $\calH^{\prime}$, will appear in the above product at least $n-\phi$ times.  Let $S_{\calH^{\prime}}$ and ${\bf B}$ be the source component size and the adjacency matrix, respectively, of $\calH^{\prime}$. In addition, let $p^{\prime}\triangleq |\calH^{\prime}|$.
Due to the existence of self-loops in the update dynamic, each ${\bf H}(t)$ has positive diagonal. In addition, $\pth{\bf B}^{n-\phi}$ contains $p^{\prime}$ nonzero columns, where $p^{\prime}\ge \max \{k^{\prime}, f+1\}$. Thus each of the $j_1, j_2, \ldots, j_{p^{\prime}}$ columns in $\prod_{t=r}^{r+\nu-1} {\bf H}(t)$ is lowered by $\ones\in \reals^{n-\phi}$ component-wise, i.e., $\pth{\prod_{t=r}^{r+\nu-1} {\bf H}(t)}_{\cdot j_i}\ge \ones$ for $i=1, \ldots, p^{\prime}$, where $\pth{\prod_{t=r}^{r+\nu-1} {\bf H}(t)}_{\cdot j_i}$ is the $j_i$--th column of $\prod_{t=r}^{r+\nu-1} {\bf H}(t)$. Therefore,
\begin{align*}
\pth{{\bf \Phi}(r+\nu-1,r)}_{\cdot j_i}\ge \beta^{\nu}\ones,
\end{align*}
for $i=1, 2, \ldots, p^{\prime}$, where $p^{\prime}\ge \max \{k^{\prime}, f+1\}$--noting that $\pth{{\bf \Phi}(r+\nu-1,r)}_{\cdot j_i}$ is the $j_i$--th column of ${\bf \Phi}(r+\nu-1,r)$.

\raggedleft $\square$
\end{proof}

\subsection*{Proof of Lemma \ref{lblimiting}}
\begin{proof}
From Lemma \ref{lb}, we know that there are at least $\max \{sp({\bf A}), f+1\}$ columns of ${\bf \Phi}(r+\nu-1,r)$ that are lower bounded by $\beta^{\nu}$ for a given $r$. Let $\calI_r$ be the collection of column indices such that for each $i\in \calI_r$,
\[
\pth{{\bf \Phi}(r+\nu-1,r)}_{\cdot i}\ge \beta^{\nu}\ones.
\]
Let $t>r+\nu$. From (\ref{mixing}), we know that

\begin{align*}
\lim_{t\ge r,t\diverge}{\bf \Phi}(t, r)=\ones \pi(r)^{\prime},
\end{align*}
for all $r$. By the definition of ${\bf \Phi}(t, r)$, we know for $t\ge r+\nu-1$, we have
\begin{align*}
{\bf \Phi}(t, r)={\bf \Phi}\pth{t, r+\nu}{\bf \Phi}\pth{r+\nu-1, r}.
\end{align*}
Thus,
\begin{align*}
\ones \pi(r)^{\prime}&=\lim_{t\ge r,~t\diverge}{\bf \Phi}(t, r)\\
&=\lim_{t\ge r,~t\diverge}{\bf \Phi}\pth{t, r+\nu}{\bf \Phi}\pth{r+\nu-1, r}\\
&=\pth{\ones \pi(r+\nu)^{\prime}}{\bf \Phi}\pth{r+\nu-1, r}.
\end{align*}
Thus, for each $i\in \calI_r$,
\begin{align*}
\pi_i(r)&=\sum_{j=1}^{n-\phi}\pi_{j}\pth{r+\nu}{\bf \Phi}_{ji}\pth{r+\nu-1, r}\\
&\ge \pth{\sum_{j=1}^{n-\phi}\pi_{j}\pth{r+\nu}}\beta^{\nu}\\
&\ge \beta^{\nu}.
\end{align*}

\raggedleft $\square$
\end{proof}

\section{Convergence Analysis of Algorithm 1}
\label{app:ConvergenceAlgorithm}

\subsection*{Proof of Lemma \ref{BasicIter}}
The proof of Lemma \ref{BasicIter} can be found in \cite{Nedic2009}. We present the proof below for completeness.
\begin{proof}
For any $x\in \reals$ and any $t\ge 0$,
\begin{align*}
| y(t+1)-x |^2&=|y(t)-\alpha(t) \iprod{\pi(t+1)}{{\bf d}(t)}-x |^2\\
&=| y(t)-x |^2-2\alpha(t)  \iprod{ \pi(t+1)}{ {\bf d}(t)} (y(t)-x)+\alpha^2(t)|\iprod{\pi(t+1)}{{\bf d}(t)}|^2\\
&\overset{(a)}{\le}  |y(t)-x |^2-2\alpha(t)  \iprod{ \pi(t+1)}{{\bf d}(t)} (y(t)-x)+\alpha^2(t)\norm{ \pi(t+1)}^2  \norm{{\bf d}(t)}^2\\
& \overset{(b)}{\le}|y(t)-x |^2-2\alpha(t)  \iprod{ \pi(t+1)}{{\bf d}(t)} (y(t)-x)+\alpha^2(t)  \norm{{\bf d}(t)}^2\\
&=|y(t)-x |^2-2\alpha(t) \sum_{j=1}^{n-\phi} \pi_j(t+1) d_j(t)\pth{y(t)-x}+\alpha^2(t)  \sum_{j=1}^{n-\phi}d_j^2(t).
\end{align*}
Inequality $(a)$ follows from Cauchy-Schwarz inequality. Inequality $(b)$ follows because
$$\norm{\pi(t+1)}^2=\sum_{j=1}^{n-\phi}{\pi}_j^2(t+1)\le \sum_{j=1}^{n-\phi}{\pi}_j (t+1)=1.$$
We now consider the term $d_j(t)\pth{y(t)-x}$ for any $j\in \calV-\calF$, for which we have
\begin{align}
d_j(t)\pth{y(t)-x}&=d_j(t)\pth{y(t)-x_j(t)+x_j(t)-x}\nonumber\\
&=d_j(t)\pth{y(t)-x_j(t)}+d_j(t)\pth{x_j(t)-x}\nonumber\\
&\ge -| d_j(t)|\,| y(t)-x_j(t)|+d_j(t)\pth{x_j(t)-x}\nonumber\\
&\ge -|d_j(t)|\,|y(t)-x_j(t)|+g_j\pth{x_j(t)}-g_j(x),
\label{step1}
\end{align}
since $d_j(t)$ is a gradient of $g_j(\cdot)$ at $x_j(t)$. Furthermore, by using a gradient $\delta_j(t)$ of $g_j(\cdot)$ at $y(t)$, we also have for any $j\in \calV-\calF$ and $x\in \reals$,
\begin{align}
g_j\pth{x_j(t)}-g_j(x)&=g_j\pth{x_j(t)}-g_j\pth{y(t)}+g_j\pth{y(t)}-g_j(x)\nonumber\\
&\ge \delta_j(t)\pth{x_j(t)-y(t)}+g_j\pth{y(t)}-g_j(x)\nonumber\\
&\ge -|\delta_j(t)|\, |x_j(t)-y(t)| +g_j\pth{y(t)}-g_j(x).
\label{step2}
\end{align}
Combining (\ref{step1}) and (\ref{step2}) together, it follows that for any $j\in \calV-\calF$ and any $x\in \reals$, we obtain
\begin{align*}
d_j(t)\pth{y(t)-x}\ge -(| d_j(t)|+|\delta_j(t)| ) |y(t)-x_j(t)|+g_j(y(t))-g_j(x).
\end{align*}
Therefore,
\begin{align*}
|y(t+1)-x|^2
&\le | y(t)-x |^2-2\alpha(t) \sum_{j=1}^{n-\phi} \pi_j(t+1) d_j(t)\pth{y(t)-x}+\alpha^2(t)  \sum_{j=1}^{n-\phi}d_j^2(t)\\
&\le | y(t)-x |^2\\
&\quad+2\alpha(t) \sum_{j=1}^{n-\phi} \pi_j(t+1) \pth{\pth{|d_j(t)|+|\delta_j(t)|} |y(t)-x_j(t)|-\pth{g_j\pth{y(t)}-g_j(x)}}\\
&\quad+\alpha^2(t)  \sum_{j=1}^{n-\phi}d_j^2(t)\\
&\le |y(t)-x|^2+2\alpha(t) \sum_{j=1}^{n-\phi} \pi_j(t+1) \pth{\pth{|d_j(t)|+|\delta_j(t)|} |y(t)-x_j(t)|}\\
&\quad-2\alpha(t) \sum_{j=1}^{n-\phi} \pi_j(t+1)\pth{g_j\pth{y(t)}-g_j(x)}+\alpha^2(t)  \sum_{j=1}^{n-\phi}d_j^2(t)\\
&\le |  y(t)-x|^2+4L\alpha(t) \sum_{j=1}^{n-\phi} \pi_j(t+1)|y(t)-x_j(t)|\\
&\quad-2\alpha(t) \sum_{j=1}^{n-\phi} \pi_j(t+1)\pth{g_j\pth{y(t)}-g_j(x)}+\alpha^2(t)(n-\phi)L^2.
\end{align*}
The last inequality holds from the fact that $g_j(\cdot)$ is $L$-Lipschitz continuous for each $j\in \calV$.

\raggedleft $\square$
\end{proof}

\subsection*{Proof of Lemma \ref{uub}}

\begin{proof}
Recall (\ref{updates}). For $t>0$,
\begin{align*}
{\bf x}(t)=
{\bf \Phi} (t-1, 0){\bf x}(0)-\sum_{r=1}^{t}\alpha(r-1){\bf \Phi} (t-1, r){\bf d}(r-1)
\end{align*}
then each $x_i(t)$ can be written as
\begin{align*}
x_i(t)=\sum_{j=1}^{n-\phi}{\bf \Phi}_{ij} (t-1, 0)x_{j}(0)-\sum_{r=1}^{t}\pth{\alpha(r-1)\sum_{j=1}^{n-\phi}{\bf \Phi}_{ij} (t-1, r)d_{j}(r-1)};
\end{align*}
and (\ref{yupdate}) implies that $y(t)=\sum_{j=1}^{n-\phi}\pi_{j} (0)x_{j}(0)-\sum_{r=1}^{t}\alpha(r-1) \sum_{j=1}^{n-\phi}\pi_j(r)d_j(r-1)$. Thus

\begin{align}
\label{uniformbd}
\nonumber
&|y(t)-x_i(t)|\\
&\le \left |\sum_{j=1}^{n-\phi}\pth{\pi_{j} (0)-{\bf \Phi}_{ij} (t-1, 0)}x_{j}(0)\right |+\left |\sum_{r=1}^{t}\pth{\alpha(r-1)\sum_{j=1}^{n-\phi}\pth{{\bf \Phi}_{ij} (t-1, r)-\pi_{j}(r)}d_{j}(r-1)}\right |.
\end{align}
We bound the two terms in (\ref{uniformbd}) separately. 
%
%
%
For the first term in (\ref{uniformbd}), we have
\begin{align}
\label{RHS1}
\nonumber
\left |\sum_{j=1}^{n-\phi}\pth{\pi_{j} (0)-{\bf \Phi}_{ij} (t-1, 0)}x_{j}(0)\right |&\le \sum_{j=1}^{n-\phi}\left |\pi_{j} (0)-{\bf \Phi}_{ij} (t-1, 0)\right |\, |x_{j}(0)|\\
\nonumber
&\overset{(a)}{\le} \sum_{j=1}^{n-\phi}\gamma^{\lceil \frac{t}{\nu}\rceil}\max \{|u|, |U|\}\\
&=\pth{n-\phi}\max \{|u|, |U|\}\gamma^{\lceil \frac{t}{\nu}\rceil},
\end{align}
where inequality (a) follows from Theorem \ref{convergencerate}.

In addition, the second term in (\ref{uniformbd}) can be bounded as follows.
\begin{align}
\nonumber
&\left |\sum_{r=1}^{t}\pth{\alpha(r-1)\sum_{j=1}^{n-\phi}\pth{{\bf \Phi}_{ij} (t-1, r)-\pi_{j}(r)}d_{j}(r-1)}\right |\\
\nonumber
&\quad  \overset{(a)}{\le}  \sum_{r=1}^{t-1}\pth{\alpha(r-1)\sum_{j=1}^{n-\phi}\left |{\bf \Phi}_{ij} (t-1, r)-\pi_{j}(r)\right |\,  \left |d_{j}(r-1)\right |}+\alpha(t-1) \left |d_i(t-1)-\sum_{j=1}^{n-\phi}\pi_j(t)d_j(t-1)\right |\\
\nonumber
&\quad  \le \sum_{r=1}^{t-1}\pth{\alpha(r-1)\sum_{j=1}^{n-\phi}\left |{\bf \Phi}_{ij} (t-1, r)-\pi_{j}(r)\right |\, \left |d_{j}(r-1)\right |}+\alpha(t-1)\sum_{j=1}^{n-\phi}\pi_j(t)\left |d_i(t-1)-d_j(t-1)\right |\\
\nonumber
&\quad  \le \sum_{r=1}^{t-1}\pth{\alpha(r-1)\sum_{j=1}^{n-\phi}\left |{\bf \Phi}_{ij} (t-1, r)-\pi_{j}(r)\right |}L+2\alpha(t-1) L\\
&\quad \le \pth{n-\phi}L\sum_{r=1}^{t-1}\alpha(r-1) \gamma^{\lceil \frac{t-r}{\nu}\rceil}+2\alpha(t-1) L
\label{RHS2}
\end{align}
where inequality $(a)$ follows from the fact that ${\bf \Phi}(t-1,t)={\bf I}$. Note that when $t=1$, it holds that
$$\sum_{r=1}^{t-1}\pth{\alpha(r-1)\sum_{j=1}^{n-\phi}|{\bf \Phi}_{ij} (t-1, r)-\pi_{j}(r)|\, |d_{j}(r-1)|}=0.$$
From (\ref{RHS1}) and (\ref{RHS2}), the LHS of (\ref{uniformbd}) can be upper bounded by
\begin{align*}
|y(t)-x_i(t)|\le \pth{n-\phi}\max \{|u|, |U|\}\gamma^{\lceil \frac{t}{\nu}\rceil} +\pth{n-\phi}L\sum_{r=1}^{t-1}\alpha(r-1) \gamma^{\lceil \frac{t-r}{\nu}\rceil}+2\alpha(t-1) L.
\end{align*}
The proof is complete.

\raggedleft $\square$
\end{proof}

\subsection*{Proof of Lemma \ref{consensus}}
\begin{proof}
Recall (\ref{consensus ub}),
\begin{align*}
|y(t)-x_i(t)|\le \pth{n-\phi}\max \{|u|, |U|\}\gamma^{\lceil \frac{t}{\nu}\rceil} +\pth{n-\phi}L\sum_{r=1}^{t-1}\alpha(r-1) \gamma^{\lceil \frac{t-r}{\nu}\rceil}+2\alpha(t-1) L.
\end{align*}
Since $0<\gamma\le 1$ and $\lim_{t\diverge} \alpha(t)=0$, it is easy to see that the first term and the third term on the RHS of (\ref{consensus ub}) both converge. In particular,
$$ \lim_{t\diverge} \pth{n-\phi}\max \{|u|, |U|\}\gamma^{\lceil \frac{t}{\nu}\rceil}=0,$$
and
$$\lim_{t\diverge}2\alpha(t-1)L=0.$$

Define $$\ell(t)=\sum_{r=1}^{t}\alpha(r-1) \gamma^{\lceil \frac{t+1-r}{\nu}\rceil}.$$

For any $t\ge 1$, we have
\begin{align*}
\ell(t)&=\sum_{r=1}^{t}\alpha(r-1) \gamma^{\lceil \frac{t+1-r}{\nu}\rceil}\\
&=\sum_{r=1}^{\lceil \frac{t}{2}\rceil}\alpha(r-1) \gamma^{\lceil \frac{t+1-r}{\nu}\rceil}+\sum_{r=\lfloor \frac{t}{2}\rfloor+1}^{t}\alpha(r-1) \gamma^{\lceil \frac{t+1-r}{\nu}\rceil}\\
&\le \sum_{r=1}^{\lceil \frac{t}{2}\rceil}\alpha(r-1) \gamma^{\frac{t+1-r}{\nu}}+\sum_{r=\lceil \frac{t}{2}\rceil+1}^{t}\alpha(r-1) \gamma^{\frac{t+1-r}{\nu}}\\
&\le \sum_{r=1}^{\lceil \frac{t}{2}\rceil}\alpha(0) \gamma^{\frac{t+1-r}{\nu}}+\alpha(\lceil \frac{t}{2}\rceil)\sum_{r=\lceil \frac{t}{2}\rceil+1}^{t}\gamma^{\frac{t+1-r}{\nu}}\\
&\le \alpha(0) \frac{\gamma^{\frac{t}{2\nu}}}{1-\gamma^{\frac{1}{\nu}}} +\frac{\alpha(\lceil \frac{t}{2}\rceil)}{1-\gamma^{\frac{1}{\nu}}}.
\end{align*}

Thus, we get
$$\limsup_{t\diverge} \ell(t)\le \alpha(0)\lim_{t\diverge}\frac{\gamma^{\frac{t}{2\nu}}}{1-\gamma^{\frac{1}{\nu}}} +\lim_{t\diverge}\frac{\alpha(\lceil \frac{t}{2}\rceil)}{1-\gamma^{\frac{1}{\nu}}}=0+0=0.
 $$
Taking limit sup on both sides of (\ref{consensus ub}), we have
\begin{align*}
\limsup_{t\diverge} |y(t)-x_i(t)|&\le \lim_{t\diverge}\pth{n-\phi}\max \{|u|, |U|\}\gamma^{\lceil \frac{t}{\nu}\rceil} +\limsup_{t\diverge}\pth{n-\phi}L\ell(t-1)+\lim_{t\diverge}2\alpha(t-1) L\\
&\le 0+0+0=0.
\end{align*}
On the other hand, since $|y(t)-x_i(t)|\ge 0$ for each $t$, it holds that
$$\liminf_{t\diverge} |y(t)-x_i(t)|\ge 0.$$
Thus,
$$\limsup_{t\diverge} |y(t)-x_i(t)|\le 0\le \liminf_{t\diverge} |y(t)-x_i(t)|.$$
By definition, we know $\liminf_{t\diverge} |y(t)-x_i(t)| \le \limsup_{t\diverge} |y(t)-x_i(t)|$. \\

Therefore, the limit of $|y(t)-x_i(t)|$ exists, and
$\lim_{t\diverge} |y(t)-x_i(t)|=0.$

\raggedleft $\square$
\end{proof}

\subsection*{Proof of Lemma \ref{c1}}

\begin{proof}
Since $\sum_{j=1}^{n-\phi} \pi_i(t+1)=1$, by Lemma \ref{uub}, we have for all $i\in \calV-\calF$,
\begin{align*}
&\sum_{j=1}^{n-\phi} \pi_i(t+1)|y(t)-x_j(t)|\\
&\le  \sum_{j=1}^{n-\phi} \pi_i(t+1)\pth{\pth{n-\phi}\max \{|u|, |U|\}\gamma^{\lceil \frac{t}{\nu}\rceil} +\pth{n-\phi}L\sum_{r=1}^{t-1}\alpha(r-1) \gamma^{\lceil \frac{t-r}{\nu}\rceil}+2\alpha(t-1) L}\\
&\le \pth{n-\phi}\max \{|u|, |U|\}\gamma^{\lceil \frac{t}{\nu}\rceil} +\pth{n-\phi}L\sum_{r=1}^{t-1}\alpha(r-1) \gamma^{\lceil \frac{t-r}{\nu}\rceil}+2\alpha(t-1) L.
\end{align*}
Using the inequality that for each $r$ and $t$
 $$\alpha(t)\alpha(r-1)\le \frac{1}{2}\pth{\alpha^2(t)+\alpha^2(r-1)},$$
we obtain
\begin{align}
\nonumber
&\sum_{t=2}^{\infty}\alpha(t)\sum_{j=1}^{n-\phi} \pi_i(t+1)|y(t)-x_j(t)|\\
\nonumber
&\le \sum_{t=2}^{\infty}\pth{\alpha(t)\pth{n-\phi}\max \{|u|, |U|\}\gamma^{\lceil \frac{t}{\nu}\rceil} +\pth{n-\phi}L\sum_{r=1}^{t-1}\alpha(t)\alpha(r-1) \gamma^{\lceil \frac{t-r}{\nu}\rceil}+2\alpha(t)\alpha(t-1) L}\\
\nonumber
&\le \sum_{t=2}^{\infty}\alpha(t)\pth{n-\phi}\max \{|u|, |U|\}\gamma^{\lceil \frac{t}{\nu}\rceil} +\frac{(n-\phi)L}{2}\sum_{t=2}^{\infty}\sum_{r=1}^{t-1}\alpha^2(t)\gamma^{\lceil \frac{t-r}{\nu}\rceil}\\
&\quad +\frac{(n-\phi)L}{2}\sum_{t=2}^{\infty}\sum_{r=1}^{t-1}\alpha^2(r-1)\gamma^{\lceil \frac{t-r}{\nu}\rceil}+\sum_{t=2}^{\infty}\pth{\alpha^2(t)+\alpha^2(t-1)}L.
\label{summable}
\end{align}
To show $\sum_{t=2}^{\infty}\alpha(t)\sum_{j=1}^{n-\phi} \pi_i(t+1)|y(t)-x_j(t)|<\infty$, we show each of the four terms in the RHS of (\ref{summable}) is finite.\\

For the first term on the RHS of (\ref{summable}), we have
\begin{align}
\label{s term1}
\nonumber
\sum_{t=2}^{\infty}\alpha(t)\pth{n-\phi}\max \{|u|, |U|\}\gamma^{\lceil \frac{t}{\nu}\rceil}&=\pth{n-\phi}\max \{|u|, |U|\}\sum_{t=2}^{\infty}\alpha(t)\gamma^{\lceil \frac{t}{\nu}\rceil}\\
\nonumber
&\overset{(a)}{\le} \pth{n-\phi}\max \{|u|, |U|\}\alpha(1)\sum_{t=2}^{\infty}\gamma^{\lceil \frac{t}{\nu}\rceil}\\
\nonumber
&\le \pth{n-\phi}\max \{|u|, |U|\}\alpha(1)\sum_{t=2}^{\infty}\gamma^{\frac{t}{\nu}}\\
\nonumber
&\le \pth{n-\phi}\max \{|u|, |U|\} \frac{\alpha(1)}{1-\gamma^{\frac{1}{\nu}}}\\
&<\infty.
\end{align}
Inequality $(a)$ holds due to the fact that $\alpha(t)\le \alpha(1)$ for all $t\ge 1$. 

For the second term in the RHS of (\ref{summable}), we have
\begin{align}
\nonumber
\frac{(n-\phi)L}{2}\sum_{t=2}^{\infty}\sum_{r=1}^{t-1}\alpha^2(t)\gamma^{\lceil \frac{t-r}{\nu}\rceil}&\le\frac{(n-\phi)L}{2}\sum_{t=2}^{\infty}\sum_{r=1}^{t-1}\alpha^2(t)\gamma^{\frac{t-r}{\nu}}\\
\nonumber
&\le\frac{(n-\phi)L}{2}\sum_{t=2}^{\infty}\alpha^2(t)\sum_{r=1}^{t-1}\gamma^{\frac{t-r}{\nu}}\\
\nonumber
&\le\frac{(n-\phi)L}{2}\sum_{t=2}^{\infty}\alpha^2(t)\sum_{r=1}^{\infty} \gamma^{\frac{r}{\nu}}\\
\nonumber
&\le\frac{(n-\phi)L}{2(1-\gamma^{\frac{1}{\nu}})}\sum_{t=2}^{\infty}\alpha^2(t)~~~~~~\text{since}~~\sum_{r=1}^{\infty} \gamma^{\frac{r}{\nu}}=\frac{\gamma^{\frac{1}{\nu}}}{1-\gamma^{\frac{1}{\nu}}}\le \frac{1}{1-\gamma^{\frac{1}{\nu}}} \\
&<\infty~~~~~~~~~~~\text{due to the fact that}~\sum_{t=0}^{\infty}\alpha^2(t)<\infty
\label{s term2}
\end{align}

For the forth term in the RHS of (\ref{summable}), we get
\begin{align}
\label{s term3}
\sum_{t=2}^{\infty}\pth{\alpha^2(t)+\alpha^2(t-1)}L=
L\sum_{t=2}^{\infty}\alpha^2(t)+L\sum_{t=2}^{\infty}\alpha^2(t-1)<\infty.
\end{align}

For the third term in the RHS of (\ref{summable}), for any fixed $N$, we get
\begin{align}
\nonumber
\frac{(n-\phi)L}{2}\sum_{t=2}^{N}\sum_{r=1}^{t-1}\alpha^2(r-1)\gamma^{\lceil \frac{t-r}{\nu}\rceil}&\le \frac{(n-\phi)L}{2}\sum_{t=2}^{N}\sum_{r=1}^{t-1}\alpha^2(r-1)\gamma^{\frac{t-r}{\nu}}\\
\nonumber
&=\frac{(n-\phi)L}{2}\sum_{r=1}^{N-1}\alpha^2(r-1)\sum_{t=1}^{N-r}\gamma^{\frac{t}{\nu}}\\
\nonumber
&\le \frac{(n-\phi)L}{2(1-\gamma^{\frac{1}{\nu}})}\sum_{r=1}^{N-1}\alpha^2(r-1).
\end{align}
Thus, we get
\begin{align}
\label{s term4}
\frac{(n-\phi)L}{2}\sum_{t=2}^{\infty}\sum_{r=1}^{t-1}\alpha^2(r-1)\gamma^{\lceil \frac{t-r}{\nu}\rceil}\le \frac{(n-\phi)L}{2(1-\gamma^{\frac{1}{\nu}})}\sum_{r=1}^{\infty}\alpha^2(r-1)<\infty.
\end{align}
In addition, for $t=0$, it holds that $|y(0)-x_j(0)|\le U-u$. For $t=1$, by Lemma \ref{uub}, we have
\begin{align*}
\sum_{j=1}^{n-\phi} \pi_i(2)|y(1)-x_j(1)| \le (n-\phi)\max\{|u|, |U|\} \gamma^{\lceil \frac{1}{\nu}\rceil}+2\alpha(0)L.
\end{align*}
Thus,
\begin{align}
\nonumber
\alpha(0)\sum_{j=1}^{n-\phi} \pi_i(1)|y(0)-x_j(0)|+\alpha(1)\sum_{j=1}^{n-\phi} \pi_i(2)|y(1)-x_j(1)|&\le \alpha(0)\pth{U-u}+2\alpha(0)L\\
\nonumber
&\quad+\alpha(1)(n-\phi)\max\{|u|, |U|\} \gamma^{\lceil \frac{1}{\nu}\rceil}\\
&<\infty.
\label{t=01}
\end{align}
By (\ref{s term1}), (\ref{s term2}), (\ref{s term3}), (\ref{s term4}) and (\ref{t=01}), we conclude that
$$\sum_{t=0}^{\infty}\alpha(t)\sum_{j=1}^{n-\phi} \pi_i(t+1)|y(t)-x_j(t)|<\infty,$$
proving the lemma.

\raggedleft $\square$
\end{proof}

\section*{Proof of Theorem \ref{ConAlgorithm1}}

\begin{proof}
Recall that each $g_i(\cdot)$ is defined as
\begin{align*}
g_i(x)={\bf A}_{1i}h_1(x)+{\bf A}_{2i}h_2(x)+\ldots +{\bf A}_{ki}h_k(x),
\end{align*}
for $i\in \calV$, where ${\bf A}_{ji}\ge 0$ and $\sum_{j=1}^k{\bf A}_{ji}=1$. Let $Y^i=\argmin~ g_i(x)$ and $Y_j^i=\argmin~{\bf A}_{ji}h_j(x)$ for $j=1, \ldots, k$. Since for each $j\in \{1, \ldots, k\}$ such that ${\bf A}_{ji}=0$, $\argmin~{\bf A}_{ji}h_j(x)=0$ is a constant function over the whole real line, it holds that $Y_j^i=\reals$. Since positive constant scaling does not affect the optimal set of a function, for each $j\in \{1, \ldots, k\}$ such that ${\bf A}_{ji}>0$, it holds that $Y_j^i=X_j$.  In addition, because $h_1(x), \ldots, h_k(x)$ are solution redundant functions, i.e., $\cap_{j=1}^k X_j\not=\O$,  functions ${\bf A}_{1i}h_1(x), \ldots, {\bf A}_{ki}h_k(x)$ are also solution redundant.  By Proposition \ref{p2} we have
$$Y^i=\cap_{j: {\bf A}_{ji}>0}X_j\supseteq \cap_{j=1}^k X_j=X, ~ \text{for all}~ i\in \calV.$$
Let $x^{\prime}\in X$. Define $g_j^*$ as the optimal value of function $g_j(\cdot)$ for each $j\in \calV$. We have
\begin{align}
\label{almost monotone}
\nonumber
|y(t+1)-x^{\prime}|^2&\le |y(t)-x^{\prime}|^2+4L\alpha(t) \sum_{j=1}^{n-\phi} \pi_j(t+1)|y(t)-x_j(t)|\\
\nonumber
&\quad-2\alpha(t) \sum_{j=1}^{n-\phi} \pi_j(t+1)\pth{g_j\pth{y(t)}-g_j(x^{\prime})}+\alpha^2(t)(n-\phi)L^2\\
\nonumber
&\overset{(a)}{=}|y(t)-x^{\prime}|^2+4L\alpha(t) \sum_{j=1}^{n-\phi} \pi_j(t+1)\left|y(t)-x_j(t)\right|\\
&\quad-2\alpha(t) \sum_{j=1}^{n-\phi} \pi_j(t+1)\pth{g_j\pth{y(t)}-g_j^*}+\alpha^2(t)(n-\phi)L^2.
\end{align}
Equality $(a)$ holds because of $x^{\prime}\in X\subseteq Y^j$ for each $j\in \calV$, then $g_j(x^{\prime})=g_j^*$.\\

For each $t\ge 0$, define
\begin{align*}
a_t&=|y(t)-x^{\prime}|^2,\\
b_t&= 2\alpha(t) \sum_{j=1}^{n-\phi} \pi_j(t+1)\pth{g_j(y(t))-g_j^*},\\
c_t&=4L\alpha(t) \sum_{j=1}^{n-\phi} \pi_j(t+1)|y(t)-x_j(t)|+\alpha^2(t)(n-\phi)L^2.
\end{align*}
It is easy to see that $a_t\ge 0$ and $c_t\ge 0$ for each $t$.
Since $g_j^*$ is the optimal value of function $g_j(\cdot)$, it holds that $b_t\ge 0$ for each $t$. Thus, $\{a_t\}_{t=0}^{\infty}, \{b_t\}_{t=0}^{\infty}$ and $\{c_t\}_{t=0}^{\infty}$ are three non-negative sequences. By (\ref{almost monotone}), it holds that
$$a_{t+1}\le a_t-b_t+c_t ~~~~\text{for each}~t\ge 0.$$

By Lemma \ref{c1}, it holds that
$$4L\sum_{t=0}^{\infty}\alpha(t) \sum_{j=1}^{n-\phi} \pi_j(t+1)\,|y(t)-x_j(t)|<\infty.$$
In addition, since $\sum_{t=0}^{\infty}\alpha^2(t)<\infty$, it holds that
$$(n-\phi)L^2\sum_{t=0}^{\infty}\alpha^2(t)<\infty.$$
Thus, we get
\begin{align}
\nonumber
\sum_{t=0}^{\infty} c_t&=\sum_{t=0}^{\infty}\pth{4L\alpha(t) \sum_{j=1}^{n-\phi} \pi_j(t+1)\,|y(t)-x_j(t)|+\alpha^2(t)(n-\phi)L^2}\\
\nonumber
&=4L\sum_{t=0}^{\infty}\pth{\alpha(t) \sum_{j=1}^{n-\phi} \pi_j(t+1)\,|y(t)-x_j(t)|}+(n-\phi)L^2\sum_{t=0}^{\infty}\alpha^2(t)\\
&< \infty.
\label{almost dominant}
\end{align}
Therefore, applying Lemma \ref{stochatic convergence} to the sequences $\{a_t\}_{t=0}^{\infty}, \{b_t\}_{t=0}^{\infty}$ and $\{c_t\}_{t=0}^{\infty}$, we have that for any $x^{\prime}\in X$,
$a_t=|y(t)-x^{\prime}|$ converges, and
\begin{align}
\label{finite b}
\sum_{t=0}^{\infty}b_t=\sum_{t=0}^{\infty}\alpha(t) \sum_{j=1}^{n-\phi}\pi_j(t+1)\pth{g_j(y(t))-g_j^*}<\infty.
\end{align}
Since $|y(t)-x^{\prime}|$ converges for any fixed $x^{\prime}\in X$, by definition of sequence convergence and the dynamic of $y(t)$ in (\ref{yupdate}), it is easy to see that $y(t)$ also converges. Let $\lim_{t\diverge}y(t)=y$. Next we show that $y\in X$.\\

By continuity of $h(\cdot)$, we have $$\lim_{t\diverge}h\pth{y(t)}=h\pth{\lim_{t\diverge}y(t)}=h(y).$$
 Equivalently, for any $\epsilon>0$, there exists $T$ such that for any $t\ge T$, it holds that
\begin{align*}
|h(y(t))-h(y)|<\epsilon.
\end{align*}
Suppose $y\notin X$, then $h(y)-h^*>0$. Let $\epsilon_0=\frac{h(y)-h^*}{2}$. Then there exists $T_0$ such that for any $t\ge T_0$, it holds that
\begin{align}
\label{continu y}
|h(y(t))-h(y)|<\epsilon_0.
\end{align}

\vskip 1.5\baselineskip
Let $\calI_{t+1}\subseteq \calV-\calF$ be the set of indices such that for each $j\in \calI_{t+1}$, $\pi_j(t+1)\ge \beta^{\nu}$. As $G(\calV, \calE)$ satisfies Condition 2,  $|\calI_{t+1}|\ge \max \{k^{\prime}, f+1\}$.  Since $g_j(y(t))-g_j^*\ge 0$ for all $j$,
then
\begin{align}
\label{essential function}
\nonumber
\sum_{j=1}^{n-\phi} \pi_j(t+1)\pth{g_j(y(t))-g_j^*}&\ge \sum_{j\in \calI_{t+1}} \pi_j(t+1)\pth{g_j(y(t))-g_j^*}\\
\nonumber
&\ge \beta^{\nu}\sum_{j\in \calI_{t+1}}\pth{g_j(y(t))-g_j^*}\\
\nonumber
&=\beta^{\nu}\sum_{j\in \calI_{t+1}}\sum_{i=1}^{k}{\bf A}_{ij}\pth{h_i(y(t))-h_i^*}\\
\nonumber
&=\beta^{\nu}\sum_{i=1}^{k}\pth{\sum_{j\in \calI_{t+1}}{\bf A}_{ij}}\pth{h_i(y(t))-h_i^*}\\
&\ge k\beta^{\nu}C_2 \pth{h(y(t))-h^*},
\end{align}
where
\begin{align*}
C_2=\min_{\calI\subseteq \calV: \,  |\calI|\ge \max \{k^{\prime}, f+1\}} \sum_{i\in \calI} {\bf A}_{ij},
\end{align*}
and the last inequality follows from the fact that $h_i(y(t))-h_i^*\ge 0$. In addition, as $sp({\bf A})=k^{\prime}$, then $\sum_{i\in \calI} {\bf A}_{ij}>0$ for every $\calI\subseteq \calV: \,  |\calI|\ge \max \{k^{\prime}, f+1\}$. Since $\bf A$ is finite, $C_2$ is well-defined and $C_2>0$.
The relation (\ref{essential function}) can be further bounded as follows.
\begin{align*}
\sum_{t=0}^{\infty}\alpha(t)\sum_{j=1}^{n-\phi} \pi_j(t+1)\pth{g_j(y(t))-g_j^*}&\ge \sum_{t=0}^{\infty}\alpha(t) k\beta^{\nu}C_2 \pth{h(y(t))-h^*}\\
&\ge \sum_{t=T_0}^{\infty}\alpha(t) k\beta^{\nu}C_2 \pth{h(y(t))-h^*}~\text{as}~h(y(t))-h^*\ge 0, \forall t\\
&\ge \sum_{t=T_0}^{\infty}\alpha(t) k\beta^{\nu}C_2 \pth{h(y)-h^*-\epsilon_0}~~~~~~\text{by (\ref{continu y})}\\
&=\sum_{t=T_0}^{\infty}\alpha(t) k\beta^{\nu}C_2 \epsilon_0~~~~~~\text{since}~\epsilon_0=\frac{h(y)-h^*}{2} \\
&=\infty ~~~~~\text{since}~\sum_{t=0}^{\infty}\alpha(t)=\infty.
\end{align*}
This contradicts the fact that (\ref{finite b}) holds. Thus, the assumption that $y\notin X$ does not hold, and $y\in X$. \\

Therefore, we conclude that $y\in X$.
That is, there exists $x^*\in X$ such that $y=x^*$ and
\begin{align}
\lim_{t\diverge} |y(t)-x^*|=0.
\label{converge y}
\end{align}

By triangle inequality, we have
$$|x_i(t)-x^*|\le |x_i(t)-y(t)|+|y(t)-x^*|.$$

Then, by Lemma \ref{consensus} and (\ref{converge y}), we have
$$\limsup_{t\diverge} |x_i(t)-x^*|\le \lim_{t\diverge} |x_i(t)-y(t)|+\lim_{t\diverge} |y(t)-x^*|=0.$$
On the other hand, $\liminf_{t\diverge} |x_i(t)-x^*|\ge 0.$
Thus, limit of $|x_i(t)-x^*|$ exists and
$$\lim_{t\diverge} |x_i(t)-x^*|=0,$$
proving Theorem \ref{ConAlgorithm1}.

\raggedleft $\square$
\end{proof}

\end{document}